\documentclass[12pt]{article}

\usepackage{scicite, dsfont, amssymb, amsmath,graphicx}

\usepackage{times}

\newenvironment{sciabstract}{%
\begin{quote} \bf}
{\end{quote}}

\newcommand{\id}{\mathds{1}}
\newcommand{\bra}[1]{\langle #1|}
\newcommand{\ket}[1]{|#1\rangle}
\renewcommand{\th}{^\mathrm{th}}
\newcommand\T{\rule{0pt}{2.6ex}}
\newcommand\B{\rule[-1.2ex]{0pt}{0pt}}
\newenvironment{proof}{\noindent \textbf{Proof:}}{\hfill$\Box$}

\newcommand{\braket}[2]{\langle #1|#2\rangle}

\newcommand{\eq}[1]{Eq.~\ref{#1}}
\newcommand{\sect}[1]{\S~\ref{#1}}
\newtheorem{theorem}{Theorem}
\newtheorem{proposition}{Proposition}

\newcounter{lastnote}
\newenvironment{scilastnote}{%
\setcounter{lastnote}{\value{enumiv}}%
\addtocounter{lastnote}{+1}%
\begin{list}%
{\arabic{lastnote}.}
{\setlength{\leftmargin}{.22in}}
{\setlength{\labelsep}{.5em}}}
{\end{list}}

\title{Quantum Algorithms for Quantum Field Theories}

\author
{Stephen P.\ Jordan,$^{1}$ Keith S.\ M.\ Lee,$^{2}$ John Preskill$^{3}$\\
\\
\normalsize{$^{1}$National Institute of Standards and Technology,
  Gaithersburg, MD}\\
\normalsize{$^{2}$University of Pittsburgh, Pittsburgh, PA}\\
\normalsize{$^{3}$California Institute of Technology, Pasadena, CA}\\
}

\date{}

\begin{document} 


\baselineskip24pt


\maketitle 


\begin{sciabstract}
  Quantum field theory reconciles quantum mechanics and special
  relativity, and plays a central role in many areas of physics.
  We develop a quantum algorithm to compute relativistic
  scattering probabilities in a massive quantum field theory with
  quartic self-interactions ($\phi^4$ theory) in 
  spacetime of four and fewer dimensions. Its run time is polynomial 
  in the number of particles, their energy, and the desired
  precision, and applies at both weak and strong coupling. In the
  strong-coupling and high-precision regimes, our quantum algorithm
  achieves exponential speedup over the fastest known classical
  algorithm.
\end{sciabstract}

The question whether quantum field theories can be
efficiently simulated by quantum computers was first posed by Feynman
three decades ago when he introduced the notion of quantum
computers \cite{Feynman}. Since then, efficient quantum algorithms have
been developed to simulate the dynamics of quantum lattice models and
quantum systems with a fixed number of particles \cite{Lloyd_science,
  Abrams_Lloyd, Cleve_sim, Childs_sim, Zalka, Wiesner}, but the
question about quantum field theories has remained open. 

In this paper, we show that quantum computers can efficiently
calculate scattering probabilities in continuum $\phi^4$ theory to
an arbitrary degree of precision. We choose $\phi^4$ theory because it
is among the simplest interacting quantum field theories, and thus
illustrates the essential issues without unnecessary complications. We
simulate a process in which initially well-separated particles with
well-defined momentum scatter off each other.

In our algorithm, we introduce several new techniques. 
First, we show that the field can be accurately represented  
with finitely many qubits by discretization of space via a lattice,
and discretization of the field value at each lattice site. 
Analyzing spatial discretization
errors is highly non-trivial for quantum field theories because of
renormalization. We approach this problem using
effective field theory. Secondly, we must create the initial state for
the simulation. We do so by developing a modified version of
adiabatic state preparation suitable for preparing non-eigenstates,
such as wavepackets. Thirdly, to improve the efficiency of simulating
the time evolution, we show that Suzuki-Trotter formulae converge
faster in cases where the underlying Hamiltonians have spatial
locality. These techniques may be of independent interest,
beyond their application to simulating quantum field theory.

No previous paper has addressed the quantum computation of scattering
amplitudes or the convergence of quantum simulations to the continuum
limit of a quantum field theory. The issue of gauge symmetries in
qubit representations of lattice field theories has been
studied~\cite{Byrnes}, and there is an extensive literature on 
how experimentally to construct Hamiltonians that approximate lattice 
gauge theories, in systems of atoms or superconducting qubits (see, for
example, \cite{Lewenstein, Johanning,
  Boada, Fischer, Menicucci, Petersen, Nation, Casanova, Casanova2,
  Snoek, Snoek2}). These previous studies are on the
experimental analog implementation of lattice Hamiltonians, whereas the
present work addresses digital simulation, with explicit
consideration of convergence to the continuum, and efficient
preparation of wavepacket states for the computation of dynamical
quantities such as scattering probabilities.

The input to our algorithm is a list of the momenta of the incoming
particles, and the output is a list of the momenta of the outgoing
particles produced by the physical scattering
process. At relativistic energies, the number of outgoing particles 
may differ from the number of incoming particles. However, because we
consider only the case of non-zero particle mass, the number of
outgoing particles is at most linear in the center-of-mass energy of
the incoming particles. In accordance with quantum mechanics, the
incoming momenta do not uniquely determine the outgoing
momenta, but rather a probability distribution over
possible outcomes. Upon repeated runs, our quantum algorithm samples
from this distribution. We quantify the precision of our simulation by
demanding that the probability of a given outcome from the simulation
differ from the true physical probability by no more than $\pm
\epsilon$.

The scattering processes simulated closely match experiments in
particle accelerators, which are the standard tools to probe uniquely 
quantum field-theoretical effects. The problem of calculating the scattering 
amplitudes, encoded in an object called the $S$-matrix, has consequently 
been well studied. 

In complexity theory, the efficiency of an algorithm is judged by
how its computational demands scale with the problem size or some other
quantity associated with the problem's intrinsic difficulty.
An algorithm with polynomial-time asymptotic scaling is considered to
be feasible, whereas one with super-polynomial (typically, exponential)
scaling is considered infeasible. This classification has proved to
be a very useful guide in practice. Our results can
be roughly summarized as follows: the calculation of quantum
field-theoretical scattering amplitudes at high precision or strong
coupling is infeasible on classical computers with known techniques
but feasible on quantum computers.

Traditional calculations of QFT scattering amplitudes rely upon 
perturbation theory, namely, a series expansion in powers of the coupling
(the coefficient of the interaction term), which is taken to be small. 
A powerful and intuitive way of organizing this perturbative expansion is 
through Feynman diagrams, in which the number of loops is associated with 
the power of the coupling. A reasonable measure of the computational
complexity of perturbative calculations is therefore the number of
Feynman diagrams, 
which is determined by
combinatorics, and grows factorially with the number of loops and the
number of external particles. 

If the coupling constant is insufficiently small, the perturbation
series does not yield correct results. In $\phi^4$ theory, for
$D=2,3$ spacetime dimensions, by increasing the coupling $\lambda_0$, 
one eventually reaches a quantum phase 
transition at some critical coupling
$\lambda_c$~\cite{Glimm:1974tz,Guerra:1975ym,McBryan:1976ga}. 
In the parameter space near this phase transition, perturbative
methods become unreliable; this region is referred to as the
strong-coupling regime.  There are then no known feasible classical methods
for calculating scattering amplitudes, although lattice field theory
can be used to obtain static quantities, such as mass ratios. Even at
weak coupling, the perturbation series is not convergent, although it
is asymptotic
\cite{Osterwalder:1975zn,Eckmann:1976xa,Constantinescu:1977xr}. 
Including higher-order contributions beyond a certain point makes the
approximation worse. There is thus a maximum possible precision
achievable perturbatively. 

We find that the number of quantum gates, $G_{\mathrm{weak}}$, needed to 
sample from scattering probabilities in weakly coupled, $(d+1)$-dimensional
$\phi^4$ theory with accuracy $\pm \epsilon$ scales as
follows\footnote{$f(n) = o(g(n))$ if and only if $\lim_{n \to \infty}
f(n)/g(n) = 0$. In the case of $\epsilon$ scaling it is of course
$1/\epsilon$ that is taken to infinity. We have used little-$o$ 
notation to simplify our exposition. For more technical detail,
see Appendix~\ref{supp}.}:

\begin{equation}
\label{weakscale}
G_{\mathrm{weak}} \sim
\left\{ \begin{array}{ll}
\left( \frac{1}{\epsilon} \right)^{1.5+o(1)} \,, & d=1 \,,\\
\left( \frac{1}{\epsilon} \right)^{2.376+o(1)} \,, & d=2 \,,\\
\left( \frac{1}{\epsilon} \right)^{3.564+o(1)} \,, & d=3 \,.
\end{array} \right.
\end{equation}

The asymptotic scaling of the number of gates used to simulate
the strongly coupled theory is summarized in Table~\ref{strongtable}.

\begin{table}[hbt]
\begin{center}
\begin{tabular}{|c|c|c|c|}
\hline
    & $\lambda_c - \lambda_0$ & $p$ & $n_{\mathrm{out}}$ \\
\hline
$d=1$ & $\left( \frac{1}{\lambda_c - \lambda_0} \right)^{8+o(1)}$ & 
$p^{4+o(1)}$ & $\tilde{O}(n_{\mathrm{out}}^5)$ \\
\hline
$d=2$ & $\left( \frac{1}{\lambda_c-\lambda_0} \right)^{5.04+o(1)}$ &
$p^{6+o(1)}$ & $\tilde{O}(n_{\mathrm{out}}^{7.128})$ \\
\hline
\end{tabular}
\end{center}
\caption{\label{strongtable} The asymptotic scaling
  of the number of quantum gates needed to simulate scattering in the
  strong-coupling regime in one and two spatial dimensions is polynomial
  in $p$, the momentum of the incoming pair of particles,
  $\lambda_c - \lambda_0$, the distance from the 
  phase transition, and $n_{\mathrm{out}}$, the maximum kinematically
  allowed number of outgoing particles. The notation
  $f(n) = \tilde{O}(g(n))$ means $f(n) = O(g(n) \log^c(n))$ for some
  constant $c$.}
\end{table}

Although quantum field theory is typically expressed in terms of
Lagrangians, and within the interaction picture, our algorithm is 
more naturally described in the formalism of Hamiltonians, and within the 
Schr\"odinger picture.
We start by defining a lattice $\phi^4$ theory,
and subsequently address convergence to the continuum theory.  
(In $D=4$, the continuum limit is believed to be the free theory.
Nevertheless, since the coupling shrinks only logarithmically, 
scattering processes for particles with small momenta in lattice units 
are still interesting to compute.)
Let
$\Omega = a \mathbb{Z}_{\hat{L}}^d$, that is, an $\hat{L} \times
\ldots \times \hat{L}$ lattice in $d$ spatial dimensions with periodic
boundary conditions and lattice spacing $a$. The number of lattice
sites is $\mathcal{V} = \hat{L}^d$. For each
$\mathbf{x} \in \Omega$, let $\phi(\mathbf{x})$ be a continuous, real
degree of freedom --- interpreted as the field at $\mathbf{x}$ ---
and $\pi(\mathbf{x})$ the corresponding canonically conjugate variable. 
In canonical quantization, these degrees of freedom are promoted to 
Hermitian operators with the commutation relation
\begin{equation}
\label{canonical}
[ \phi(\mathbf{x}), \pi(\mathbf{y})] = i a^{-d} \delta_{\mathbf{x},
  \mathbf{y}} \id.
\end{equation}
As is standard in quantum field theory, we use units with $\hbar = c =1$. 
$\phi^4$ theory on the lattice $\Omega$ is defined by the Hamiltonian
\begin{equation}
\label{Hamiltonian}
H = \sum_{\mathbf{x} \in \Omega} a^d \left[ \frac{1}{2}
  \pi(\mathbf{x})^2 + \frac{1}{2} ( \nabla_a \phi)^2(\mathbf{x}) +
  \frac{1}{2} m_0^2 \phi(\mathbf{x})^2 + \frac{\lambda_0}{4!}
  \phi(\mathbf{x})^4 \right],
\end{equation}
where $\nabla_a \phi$ denotes a discretized derivative, that is,
a finite-difference operator. 

We represent the state of the lattice field theory by devoting one
register of qubits to store the value of the field at each lattice
point. Each $\phi(\mathbf{x})$ is in principle an unbounded continuous
variable. To represent the field at a given site with finitely many
qubits, we cut off the field at a maximum magnitude $\phi_{\max}$ and
discretize it in increments of $\delta_{\phi}$. This requires $n_b =
O(\log(\phi_{\max}/\delta_{\phi}))$ qubits per site. Note that this
field discretization is a separate issue from the spatial discretization
via the lattice $\Omega$.

Let $\ket{\psi}$ be any state such that $\bra{\psi} H \ket{\psi} \leq
E$. The probability distribution over $\phi(\mathbf{x})$ defined by
$\ket{\psi}$ (for any $\mathbf{x} \in \Omega$) has a very low
probability\footnote{For $\lambda_0 > 0$ one has a tighter bound. In this
  case it is unlikely for $|\phi(\mathbf{x})|$ to be much larger than
  $O(E^{1/4})$ (\S \ref{qubits}).} for $|\phi(\mathbf{x})|$ to be much larger than
$O(\sqrt{E})$. Thus, a cutoff $\phi_{\max} = O \left( \sqrt{
  \frac{\mathcal{V} E}{a^d m_0^2 \epsilon}} \right)$
suffices to ensure fidelity
$1-\epsilon$ to the original state $\ket{\psi}$. One can prove this by
bounding $\bra{\psi}\phi(\mathbf{x}) \ket{\psi}$ and $\bra{\psi}
\phi^2(\mathbf{x}) \ket{\psi}$ as functions of $E$ and applying
Chebyshev's inequality (\S \ref{qubits}). To choose $\delta_\phi$, note
that the eigenbasis of $a^d \pi(\mathbf{x})$ is the Fourier transform
of the eigenbasis of $\phi(\mathbf{x})$. Hence, discretizing
$\phi(\mathbf{x})$ in units of $\delta_{\phi}$ is equivalent to
introducing the cutoff $-\pi_{\max} \leq \pi(\mathbf{x}) \leq
\pi_{\max}$, where $\pi_{\max} = \frac{1}{a^d   \delta_{\phi}}$. By
bounding the expectations of $\pi(\mathbf{x})$ and $\pi^2(\mathbf{x})$, 
one finds that it suffices to choose $\pi_{\max} = O \left(
\sqrt{\frac{\mathcal{V}E}{\epsilon a^d}} \right)$, and thus  $n_b = O
\left( \log \left(\frac{\mathcal{V}E}{m_0 \epsilon} \right) \right)$.

We now turn to the main three tasks of quantum simulation: preparing
an initial state, simulating the time evolution $e^{-iHt}$, and
measuring final observables. We discuss simulation of time evolution
first, as it is used in all three tasks. The unitary operator $e^{-iHt}$ can be
approximated by a quantum circuit of $O((t \mathcal{V})^{1+1/2k})$
gates implementing a $k\th$-order Suzuki-Trotter formula of the type
described in \cite{Suzuki, Cleve_sim}. This near-linear scaling with
$t$ has long been known. 
The scaling with
$\mathcal{V}$ is a consequence of the locality\footnote{$H$ couples
only nearest-neighbor sites, via the $(\nabla_a \phi)^2$ term.} of $H$
(\S \ref{Trotter}) and appears not to have been noted previously in the
quantum algorithms literature.

To simulate scattering,
one needs to prepare an initial state of particles in well-separated
wavepackets. We do so by preparing the vacuum of the $\lambda_0 = 0$
theory, exciting wavepackets, and then adiabatically turning on the
coupling $\lambda_0$.  
Let $H^{(0)}$ be the Hamiltonian obtained by setting $\lambda_0 = 0$
in $H$. $H^{(0)}$ defines an exactly solvable model in which the
particles are non-interacting. The vacuum (ground) state
$\ket{\mathrm{vac}(0)}$ of $H^{(0)}$ is a multivariate Gaussian 
wavefunction in the variables
$\{\phi(\mathbf{x})| \mathbf{x} \in \Omega\}$, and can therefore be
prepared using the method of Kitaev and Webb \cite{Kitaev_Webb}. The
asymptotic scaling of the Kitaev-Webb method is dictated by the
computation of the $\mathbf{L}\mathbf{D}\mathbf{L}^T$ decomposition of 
the covariance matrix, which can be done classically in 
$O(\mathcal{V}^{2.376})$ time with \cite{Bunch, Coppersmith}.

In analogy with the familiar case of the harmonic oscillator,
one can define creation and
annihilation operators  $a_{\mathbf{p}}$ and $a_{\mathbf{p}}^\dag$  such
that $H^{(0)} = \sum_{\mathbf{p} \in   \Gamma} L^{-d}
\omega_{\mathbf{p}} a_{\mathbf{p}}^\dag a_{\mathbf{p}} + E^{(0)} \id$,
where $\Gamma = \frac{2 \pi}{\hat{L}a} \mathbb{Z}^d_{\hat{L}}$ is the
momentum-space lattice corresponding to $\Omega$, $\omega_{\mathbf{p}}
= \sqrt{ m_0^2 + \frac{4}{a^2} \sum_{j=1}^d \sin^2 \left(
  \frac{a p_j}{2} \right)}$, and $E^{(0)}$ is an irrelevant zero-point
energy. The operator $a_{\mathbf{p}}^\dag$ can be interpreted as
creating a (completely delocalized) particle of the non-interacting
theory with momentum $\mathbf{p}$ and energy $\omega_{\mathbf{p}}$.

The (unnormalized) state $\phi(\mathbf{x}) \ket{\mathrm{vac}(0)}$ is
interpreted as a single particle localized at $\mathbf{x}$ (see, e.g.,
\cite{Peskin}).  Because $a_{\mathbf{p}} \ket{\mathrm{vac}(0)} = 0$,
$\phi(\mathbf{x}) \ket{\mathrm{vac}(0)} = a_{\mathbf{x}}^\dag
\ket{\mathrm{vac}(0)}$, where 
\begin{equation}
\label{axdag}
a_{\mathbf{x}}^\dag = \sum_{\mathbf{p} \in \Gamma} L^{-d} e^{-i
  \mathbf{p} \cdot \mathbf{x}} \sqrt{\frac{1}{2 \omega(\mathbf{p})}}
a_{\mathbf{p}}^\dag.
\end{equation}

The operator
\begin{equation}
\label{apsidag}
a_{\psi}^\dag = \eta(\psi) \sum_{\mathbf{x} \in  \Omega} a^d
\psi(\mathbf{x}) a_{\mathbf{x}}^\dag
\end{equation}
creates a wavepacket with position-space wavefunction
$\psi$. ($\eta(\psi)$ is a normalization constant, chosen so that
$[a_\psi, a_\psi^\dagger] = 1$.)
$a_{\psi}^\dag$ is not unitary, so it cannot be directly implemented
by a quantum circuit. Instead, we introduce an ancillary qubit and let
\begin{equation}
H_{\psi} = a_{\psi}^\dag \otimes \ket{1} \bra{0} + a_{\psi} \otimes
\ket{0} \bra{1}.
\end{equation}
One can verify that $e^{-i H_{\psi} \pi/2} \ket{\mathrm{vac}(0)} \ket{0} = -i
a_{\psi}^\dag \ket{\mathrm{vac}(0)} \ket{1}$. Using a high-order
Suzuki-Trotter formula \cite{Suzuki, Cleve_sim}, we can construct
an efficient quantum circuit approximating the unitary transformation
$e^{-i H_{\psi} \pi/2}$. Applied to $\ket{\mathrm{vac}(0)}$, this
circuit yields the desired state up to an irrelevant global phase and
an unentangled ancillary qubit, which can be discarded. We repeat this
process for each incoming particle desired. 

Because we wish to create localized wavepackets, we can choose
$\psi(\mathbf{x})$ to have bounded support. Expanding $a_{\psi}^\dag$
in terms of the operators $\phi$ and $\pi$ yields an expression of the 
form $a_{\psi}^\dag =
\sum_{\mathbf{x} \in \Omega} \left[ f(\mathbf{x}) \phi(\mathbf{x}) +
  g(\mathbf{x}) \pi(\mathbf{x}) \right]$, where $f(\mathbf{x})$ and
$g(\mathbf{x})$ are exponentially decaying with 
characteristic length scale $1/m_0$ outside the support of
$\psi$. Thus, $a_{\psi}$ and $a^\dag_{\psi}$ can be exponentially well
approximated by linear combinations of the operators $\phi$ and $\pi$ 
on a local region of space, and the complexity of simulating $e^{-i
  H_{\psi} \pi/2}$ does not scale with the volume $V$. Furthermore,
provided the initial wavepackets are separated by a distance that is large
compared with $1/m_0$, the preparation of each additional wavepacket
leaves the existing wavepackets almost perfectly undisturbed. 

At this point, we have finished constructing wavepackets of the
non-interacting theory. We next use a Suzuki-Trotter formula to
construct a quantum circuit simulating the unitary transformation
induced by a time-dependent Hamiltonian in which the coupling constant
is gradually increased from zero to its final value, $\lambda_0$. By
the adiabatic theorem, sufficiently slow turn-on ensures that no stray
particles are created during this process, provided particle creation
costs energy, that is, the particles have non-zero mass. In the free
theory, the particle mass is $m_0$. In the interacting theory, with
fixed $m_0$ and sufficiently large $\lambda_0$, the mass
vanishes. This marks the location of the $\phi \to -\phi$
symmetry-breaking transition. In this paper we restrict our attention
to simulations within the symmetric phase, although we do consider
systems arbitrarily close to the phase transition, as these should be
particularly hard to simulate classically.

As Eq.~\ref{apsidag} shows, wavepackets are not eigenstates of
$H^{(0)}$. During the adiabatic turn-on, the different eigenstates acquire
different dynamical phases. Thus, as the wavepacket time evolves, 
it propagates and broadens. This behavior is undesirable
in our simulation, because we do not wish the particles to collide and
scatter before the coupling reaches its final value. We therefore
introduce backward time evolutions governed by time-independent
Hamiltonians into the adiabatic state-preparation process to undo the
dynamical phases. Specifically, let $H(s)$ parameterize the adiabatic
time evolution, with $H(0) = H^{(0)}$ and $H(1) = H$. We divide the
adiabatic preparation into $J$ steps, with $U_j$ denoting the unitary
time evolution induced by the time-dependent Hamiltonian linearly
interpolating between $H((j-1)/J)$ and $H(j/J)$ over a period of
$\tau/J$. Let 
$M_j$ consist of backward, forward, and backward evolutions,
namely,
\begin{equation}
M_j = \exp \left[iH \left(\frac{j+1}{J} \right) \frac{\tau}{2J}
  \right] U_j  \exp \left[iH \left( \frac{j}{J} \right)
  \frac{\tau}{2J} \right].
\end{equation}
Our full state-preparation process is $\prod_{j=1}^J M_j$. The
dynamical phases converge to zero as $J \to \infty$, while the
adiabatic change of eigenbasis is undisturbed (\S \ref{preparing}).

After the system has evolved for a period in which scattering
occurs, measurement is performed as follows. The interaction is
adiabatically turned off, through the time-reversed version of the 
turn-on described above. Once we return to the free theory,
we can measure the number operators of the momentum modes,
using the method of phase estimation, that is, by simulating 
$e^{i L^{-d} a_{\mathbf{p}}^\dag a_{\mathbf{p}} t}$ for various 
values of $t$ and Fourier transforming the results~\cite{Kitaev95}.


Having described how, once discretized, a quantum field theory becomes 
essentially an ordinary many-body quantum-mechanical system, whose evolution 
can be efficiently simulated on quantum computers by combining
established primitives, we now consider discretization errors. 
To analyse the errors introduced to our simulation by discretization, we 
use methods of effective field theory, a well-developed formalism underlying 
our modern understanding of quantum field theory. 

In its regime of validity,
typically below a particular energy scale, an effective field theory (EFT)
reproduces the behavior of the full (that is, fundamental) theory under
consideration: it can be regarded as the low-energy limit of that theory. 
An EFT for a full theory is thus somewhat analogous to a Taylor series for
a function. 
It involves an expansion in some suitable small parameter, so that,
although it consists of infinitely many terms, higher-order terms are
increasingly suppressed. Thus, the series can be truncated, with 
corresponding finite and controllable errors.


We apply this framework to analyse the effect of discretizing the spatial
dimensions of the continuum $\phi^4$ quantum field theory.
The discretized Lagrangian can be thought of as the leading 
contribution (denoted by ${\cal L}^{(0)}$) to an effective field theory. 
From the leading operators
left out we can thus infer the scaling of the error associated with  
a non-zero lattice spacing, $a$.

The full (untruncated) effective Lagrangian will have every coupling respecting 
the $\phi \rightarrow -\phi$ symmetry, and so will take the form
\begin{equation}
{\cal L}_{\rm eff} = {\cal L}^{(0)} + \frac{c}{6!}\phi^6 
+ c'\phi^3\partial^2\phi + \frac{c''}{8!}\phi^8 + \cdots \,.
\end{equation}
This can be simplified. First, the chain rule and integration by parts
(with boundary terms dropped) can be used to write any operator with two 
derivatives acting on different fields in the form $\phi^n \partial^2 \phi$. 
For example,
$\phi^2 \partial_\mu \phi \partial^\mu \phi
 =  \frac{1}{3} \partial_\mu(\phi^3) \partial^\mu \phi
 \rightarrow  - \frac{1}{3} \phi^3 \partial^2 \phi $ .
Such an operator can then be simplified via the 
equation of motion \cite{Arzt:1993gz,Georgi:1991ch}.
If this were the equation of motion of the continuum theory, any
derivative operator would then be completely eliminated. In the 
discretized theory, however, the equation of motion is modified
and there are residual, Lorentz-violating operators.
In fact, because the difference operators in the discretized theory
are only approximately equal to the derivatives in the continuum theory,
the simplest Lorentz-violating operators are induced purely by
discretization.

In units where $\hbar = c = 1$, all quantities have units of some power of
mass. 
The mass dimensions (denoted by $[.]$) of the field and coupling in 
$D=d+1$ spacetime dimensions are
$ \left[ \phi \right] = \frac{D-2}{2}$ and $\left[ \lambda \right] = 4-D$, 
which imply that
\begin{equation}
 \left[ c \right] = 6-2D \,, \,\,\,
 \left[ c'' \right] = 8-3D \,.
\end{equation}
In $D=4$ dimensions, $\left[ c \right] = -2$ and 
$\left[ c'' \right] = -4$. Since the only relevant dimensionful
parameter is the lattice spacing, that is, \ $\Lambda \sim \pi/a$,
this means that $c \sim a^2$ and $c'' \sim a^4$. We see then
that, of the operators not included in the Lagrangian
${\cal L}^{(0)}$, $\phi^6$ is more significant than 
$\phi^{2n}$, for $n > 3$. 

In $D=2,3$, the scaling of the coefficients with $a$ is somewhat 
less obvious, because now the coupling $\lambda$ provides another 
dimensionful parameter. 
To obtain the scaling of $c$, one should consider the Feynman diagram 
that generates the corresponding operator. This involves three $\phi^4$ 
vertices, so
\begin{eqnarray}
\begin{array}{l} \includegraphics[width=0.6in]{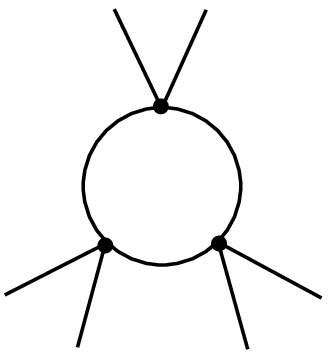} 
\end{array} 
& \sim & \lambda^3 a^{6-D} \,.
\end{eqnarray}
(Other diagrams involve higher powers of $\lambda$ and hence are
suppressed by higher powers of $a$.)
Likewise, the coefficient of $\phi^8$ will scale as $\lambda^4 a^{8-D}$, 
which means that it is suppressed by $a^2$ relative to the coefficient of 
$\phi^6$.

The effective field theory thus consists of three different
classes of operators: operators of the form $\phi^{2n}$,
Lorentz-violating operators arising solely from discretization effects,
and Lorentz-violating operators due to discretization and quantum effects.
These are shown with the scaling of their coefficients
in Table~\ref{table:eftshort}. At strong coupling, the operators and
their scaling remain the same at the scale of the matching of the full
theory on to the EFT, although the explicit coefficients are no longer 
calculable.  However, the running of the coefficients down to lower 
energies is determined by their so-called anomalous dimensions, which 
depend on the coupling strength. These
anomalous dimensions modify the scaling; at weak coupling the
modification is small, but at strong coupling it could be larger. (Still,
the scaling will remain polynomial.)

\noindent
\medskip
\begin{table}[hbt]
\begin{center}
\begin{tabular}{|c|c|c|}
\hline \T\B
Class & Operators & Scaling of coupling 
\\
\hline
\hline
\T\B
I & $\phi^{2n}$ ($n \geq 3$) & $\lambda^n a^{2n-D}$ \\[3pt]
\hline
\T\B
II & $\phi \partial_{\bf x}^{2l} \phi$ ($l \geq 2$) & $a^{2l-2}$ 
\\
\hline
\T
III & $\phi^{2j+1} \partial_{\bf x}^{2l} \phi$ 
& $\lambda^{j+1} a^{2j+2l+2-D}$ \\
\B & ($j\geq 1$, $l \geq2$) 
&  \\
\hline
\end{tabular}
\caption{Effective field theory operators fall into three classes.
The general operator in each class is shown, with the canonical scaling of 
its coefficient in $D$ spacetime dimensions. Here,
$\partial^{2l}_{\mathbf{x}} = \sum_{i=1}^d \partial_i^{2l}$. 
}
\label{table:eftshort}
\end{center}
\end{table}
\medskip

From Table \ref{table:eftshort}, one sees that the dominant
discretization errors scale as $a^2$ in $D=2,3,4$. (In $D=2,3$, errors
of type II dominate. In $D=4$, errors of types I and II each scale as
$a^2$.) These error terms shift scattering probabilities by $\pm
\epsilon$ where $\epsilon = O(a^2)$. Limiting $\epsilon$ determines
$a$ and $\mathcal{V} = \frac{V}{a^d}$, which then determine the
overall complexity of the quantum simulation algorithm described in
Eq.~\ref{weakscale} and Table~\ref{strongtable}. 

In conclusion, we have shown that quantum computers can efficiently
calculate scattering probabilities in $\phi^4$ theory to arbitrary
precision at both weak and strong coupling. 
Known classical algorithms take exponential time to do this
in the strong-coupling and high-precision regimes. 
In addition to establishing a new exponential quantum speedup, 
our algorithm introduces several new techniques.
These lead the way towards a quantum algorithm for simulating the 
Standard Model of particle physics, which has new features, such as
chiral fermions and gauge interactions. 
Such an algorithm would establish that, except for quantum-gravity effects, 
the standard quantum circuit model suffices to capture completely the 
computational power of our universe.



\bibliography{short_qft}

\begin{thebibliography}{10}

\bibitem{Feynman}
R.~P. Feynman, {\it International Journal of Theoretical Physics\/} {\bf 21},
  467 (1982).

\bibitem{Lloyd_science}
S.~Lloyd, {\it Science\/} {\bf 273}, 1073 (1996).

\bibitem{Abrams_Lloyd}
D.~S. Abrams, S.~Lloyd, {\it Physical Review Letters\/} {\bf 79}, 2586 (1997).

\bibitem{Cleve_sim}
D.~Berry, G.~Ahokas, R.~Cleve, B.~C. Sanders, {\it Communications in
  Mathematical Physics\/} {\bf 270}, 359 (2007). ArXiv:quant-ph/0508139.

\bibitem{Childs_sim}
A.~Childs, R.~Kothari, {\it arXiv:1003.3683\/}  (2010).

\bibitem{Zalka}
C.~Zalka, {\it Proceedings of the Royal Society of London A\/} {\bf 454}, 313
  (1998). ArXiv:quant-ph/9603026.

\bibitem{Wiesner}
S.~Wiesner, {\it arXiv:quant-ph/9603028\/}  (1996).

\bibitem{Byrnes}
T.~Byrnes, Y.~Yamamoto, {\it Physical Review A\/} {\bf 73}, 022328 (2006).
  ArXiv:quant-ph/0510027.

\bibitem{Lewenstein}
M.~Lewenstein, {\it et~al.\/}, {\it Advances in Physics\/} {\bf 56}, 243
  (2007). ArXiv:cond-mat/0606771.

\bibitem{Johanning}
M.~Johanning, A.~Var{\'o}n, C.~Wunderlich, {\it Journal of Physics B\/} {\bf
  42}, 154009 (2009). ArXiv:0905.0118.

\bibitem{Boada}
O.~Boada, A.~Celi, J.~I. Latorre, M.~Lewenstein, {\it arXiv:1010.1716\/}
  (2010).

\bibitem{Fischer}
U.~R. Fischer, R.~Sch{\"u}tzold, {\it Physical Review A\/} {\bf 70}, 063615
  (2004). ArXiv:cond-mat/0406470.

\bibitem{Menicucci}
N.~C. Menicucci, S.~J. Olson, G.~J. Milburn, {\it New Journal of Physics\/}
  {\bf 12}, 095019 (2010). ArXiv:1005.0434.

\bibitem{Petersen}
L.~Petersen, Quantum simulations in ion traps -- towards simulating the early
  expanding universe, arXiv:0704.3535 (2006). Diplomathesis.

\bibitem{Nation}
P.~D. Nation, M.~P. Blencowe, A.~J. Rimberg, E.~Buks, {\it Physical Review
  Letters\/} {\bf 103}, 087004 (2009). ArXiv:quant-ph/0904.2589.

\bibitem{Casanova}
J.~Casanova, {\it et~al.\/}, {\it arXiv:1107.5233\/}  (2011).

\bibitem{Casanova2}
J.~Casanova, A.~Mezzacapo, L.~Lamata, E.~Solano, {\it arXiv:1110.3730\/}
  (2011).

\bibitem{Snoek}
M.~Snoek, S.~Vandoren, H.~T.~C. Stoof, {\it Physical Review A\/} {\bf 74},
  033607 (2006). ArXiv:cond-mat/0604671.

\bibitem{Snoek2}
M.~Snoek, M.~Haque, S.~Vandoren, H.~T.~C. Stoof, {\it Physical Review
  Letters\/} {\bf 95}, 250401 (2005). ArXiv:cond-mat/0505055.

\bibitem{Glimm:1974tz}
J.~Glimm, A.~Jaffe, {\it Physical Review D\/} {\bf 10}, 536 (1974).

\bibitem{Guerra:1975ym}
F.~Guerra, L.~Rosen, B.~Simon, {\it Communications in Mathematical Physics\/}
  {\bf 41}, 19 (1975).

\bibitem{McBryan:1976ga}
O.~A. McBryan, J.~Rosen, {\it Communications in Mathematical Physics\/} {\bf
  51}, 97 (1976).

\bibitem{Osterwalder:1975zn}
K.~Osterwalder, R.~S{\'e}n{\'e}or, {\it Helvetica Physica Acta\/} {\bf 49}, 525
  (1976).

\bibitem{Eckmann:1976xa}
J.-P. Eckmann, H.~Epstein, J.~Fr{\"o}hlich, {\it Annales de l'institut Henri
  Poincar{\'e} (A) Physique th{\'e}orique\/} {\bf 25}, 1 (1976).

\bibitem{Constantinescu:1977xr}
F.~Constantinescu, {\it Annals of Physics\/} {\bf 108}, 37 (1977).

\bibitem{Suzuki}
M.~Suzuki, {\it Physics Letters A\/} {\bf 146}, 319 (1990).

\bibitem{Kitaev_Webb}
A.~Kitaev, W.~A. Webb, {\it arXiv:0801.0342\/}  (2008).

\bibitem{Bunch}
J.~R. Bunch, J.~E. Hopcroft, {\it Mathematics of Computation\/} {\bf 28}, 231
  (1974).

\bibitem{Coppersmith}
D.~Coppersmith, S.~Winograd, {\it Journal of Symbolic Computation\/} {\bf 9},
  251 (1990).

\bibitem{Peskin}
M.~E. Peskin, D.~V. Schroeder, {\it An Introduction to Quantum Field Theory\/}
  (Westview, 1995).

\bibitem{Kitaev95}
A.~Y. Kitaev, {\it arXiv:quant-ph/9511026\/}  (1995).

\bibitem{Arzt:1993gz}
C.~Arzt, {\it Physics Letters B\/} {\bf 342}, 189 (1995).

\bibitem{Georgi:1991ch}
H.~Georgi, {\it Nuclear Physics B\/} {\bf 361}, 339 (1991).

\bibitem{Grover_Rudolph}
L.~Grover, T.~Rudolph, {\it arXiv:quant-ph/0208112\/}  (2002).

\bibitem{LeGuillou:1977ju}
J.~C. Le~Guillou, J.~Zinn-Justin, {\it Phys. Rev. Lett.\/} {\bf 39}, 95 (1977).

\bibitem{Luscher:1987ay}
M.~Luscher, P.~Weisz, {\it Nucl. Phys.\/} {\bf B290}, 25 (1987).

\bibitem{Messiah}
A.~Messiah, {\it Quantum Mechanics\/} (Dover, 1999). (Reprint of the two volume
  edition published by Wiley, 1961-1962.).

\end{thebibliography}

\bibliographystyle{Science}


\begin{scilastnote}
\item We thank Alexey Gorshkov for helpful discussions. 
  This work was supported by NSF grant PHY-0803371, DOE grant 
  DE-FG03-92-ER40701, and NSA/ARO grant W911NF-09-1-0442.
  Much of this work was done while S.J. was at the Institute for Quantum
  Information (IQI), Caltech, supported by the Sherman Fairchild Foundation.
  K.L. was supported in part by NSF grant PHY-0854782. He is grateful for 
  the hospitality of the IQI, Caltech, during parts of this work. 
\end{scilastnote}

\clearpage

\appendix

\section{Supplementary Material}
\label{supp}

\subsection{Steps of Algorithm and Comments}
\label{details}

Our quantum algorithm works by the following sequence of steps.

\begin{enumerate}

\item \textbf{Prepare the free vacuum.}
Improving upon the efficiency of earlier, more general, state-construction 
  methods \cite{Zalka, Grover_Rudolph}, Kitaev and Webb
  developed a quantum algorithm for constructing multivariate Gaussian
  superpositions \cite{Kitaev_Webb}.
  For large $\mathcal{V}$, the dominant cost in Kitaev and Webb's
  method for producing $\mathcal{V}$-dimensional multivariate
  Gaussians is the computation of the $\mathbf{L}\mathbf{D}\mathbf{L}^T$ 
  decomposition of the inverse covariance matrix, where $\mathbf{L}$ is a 
  unit lower triangular matrix, and $\mathbf{D}$ is a diagonal matrix. 
  This can be done in $\tilde{O}(\mathcal{V}^{2.376})$ time with established 
  classical methods~\cite{Bunch, Coppersmith}. 
(The notation $f(n) = \tilde{O}(g(n))$ means $f(n) = O(g(n) \log^c(n))$ 
for some constant $c$.)
  The computation of the matrix 
  elements of the covariance matrix itself is easy because, for large $V$, 
  the sum
\begin{equation}
\label{freepropagator}
G^{(0)}(\mathbf{x} - \mathbf{y}) = \sum_{\mathbf{p} \in \Gamma} L^{-d}
\frac{1}{2 \omega(\mathbf{p})} e^{i \mathbf{p} \cdot (\mathbf{x}_i -
  \mathbf{x}_j)} 
\end{equation}
  defining the propagator of the lattice theory is well approximated
  by an easily evaluated integral.

\item \textbf{Excite wavepackets.}
The span of $\ket{\mathrm{vac}(0)} \ket{0}$ and
$\ket{\psi} \ket{1}$ is an invariant subspace, on which $H_\psi$
acts as
\begin{eqnarray}
H_\psi \ket{\mathrm{vac}(0)} \ket{0} & = & \ket{\psi} \ket{1} \,, \\
H_\psi \ket{\psi} \ket{1} & = & \ket{\mathrm{vac}(0)} \ket{0} \,.
\end{eqnarray}
Thus,
\begin{equation}
e^{-i H_\psi \pi/2} \ket{\mathrm{vac}(0)} \ket{0} = -i \ket{\psi} \ket{1} \,.
\end{equation}
Hence, by simulating a time evolution according to the Hamiltonian
$H_\psi$, we obtain the desired wavepacket state $\ket{\psi}$, up to an
irrelevant global phase and extra qubit, which can be discarded. After
rewriting $H_\psi$ in terms of the operators $\phi(\mathbf{x})$ and 
$\pi(\mathbf{x})$, one sees that simulating $H_\psi$ is a very similar task 
to simulating $H$, and can be done with the same techniques. 

The only errors introduced at this step are due to the finite separation
distance $\delta$ between wavepackets, and are of order $\epsilon \sim
e^{-\delta/m}$. 
(However, our wavepackets have a constant spread in momentum, and thus
differ from the idealization of particles with precisely defined momenta.)
The wavepacket preparation thus has complexity scaling
linearly with $n_{\mathrm{in}}$, the number of particles being
prepared, and necessitates a dependence $V \sim n_{\mathrm{in}}
\log(1/\epsilon)$.

\item \textbf{Adiabatically turn on the interaction.} 
For $0 \leq s \leq 1$, let 
\begin{equation}
\label{HS}
H(s) = \sum_{\mathbf{x} \in \Omega} a^d \left[ \frac{1}{2} \pi(\mathbf{x})^2 
+ \frac{1}{2} (\nabla_a \phi)^2 (\mathbf{x}) + \frac{1}{2} m_0^2(s)
  \phi(\mathbf{x})^2 + \frac{\lambda_0(s)}{4!} \phi(\mathbf{x})^4 \right]
\end{equation}
with $\lambda_0(0) = 0$.
$U_j$ is the unitary time evolution induced by $H(t/\tau)$ from
$t=\frac{j\tau}{J}$ to $t=\frac{(j+1)\tau}{J}$, namely,
\begin{eqnarray}
U_j & = & T \left\{ \exp \left[ -i \int_{j/J}^{(j+1)/J} H(s) \tau ds
  \right] \right\} \,,
\end{eqnarray}
where $T\{ \cdot \}$ indicates the time-ordered product. 
We suppress the dynamical phases by choosing $J$ to be sufficiently large.
The choice of a suitable ``path'' $\lambda_0(s),m_0^2(s)$, and
the complexity of this state-preparation process depends in a
complicated manner on the parameters in $H$ (\sect{preparing}). 

\item \textbf{Simulate Hamiltonian time evolution.}

\item \textbf{Adiabatically turn off the interaction.} 
The adiabatic turn-off of the coupling is simply the
  time-reversed version of the adiabatic turn-on. 

\item \textbf{Measure occupation numbers of momentum modes.} 
  For a given $\mathbf{p}$, measurement of  
  $L^{-d} a_{\mathbf{p}}^\dag a_{\mathbf{p}}$ by phase estimation
  can be implemented with $O \left( \mathcal{V}^{2+\frac{1}{2k}} \right)$ 
  quantum gates via a $k\th$-order Suzuki-Trotter formula. 
  Furthermore, if we instead simulate localized detectors, the
  computational cost becomes independent of $V$ (much as the
  computational cost of creating local wavepackets is independent of
  $V$), but the momentum resolution becomes lower, as dictated by the
  uncertainty principle.
\end{enumerate}

The allowable rate of adiabatic increase of the coupling constant
during state preparation is determined by the physical mass of the
theory. In the weakly coupled case, this can be calculated
perturbatively. In the strongly coupled case, such a calculation is no 
longer possible. Thus one is left with the
problem of determining how fast one can perform the adiabatic state
preparation without introducing errors. Fortunately, one can easily
calculate the mass on a quantum computer, as follows. First, one
adiabatically prepares the interacting vacuum state at some small
$\lambda_0$, and measures the energy of the vacuum using phase
estimation. 
The speed at which to increase $\lambda_0$ can be chosen
perturbatively for this small value of $\lambda_0$. Next, one
adiabatically prepares the state with a single zero-momentum particle
at the same value of $\lambda_0$, and measures its energy using phase
estimation. Subtracting these values yields the physical mass. This
value of the physical mass provides guidance as to the speed of
adiabatic increase of the coupling to reach a slightly higher
$\lambda_0$. Repeating this process for successively higher
$\lambda_0$ allows one to reach strong coupling, while always having an
estimate of mass by which to choose a safe speed for adiabatic state
preparation. In addition, mapping out the physical mass as a function
of bare parameters (hence, for example, mapping out the phase diagram)
may be of independent interest.

\subsection{Efficiency}
\label{efficiency}

To quantify the precision of a simulation, we demand that the
probability of a given scattering event in the simulation differ from
the true physical probability by no more than $\pm \epsilon$. There
are various sources of error: discretization of space, Trotter approximations,
imperfect adiabaticity, discretization and cutoff of the field at each
site, and imperfect spatial separation of particles in the asymptotic in
and out states. In a theory with a non-zero mass, errors due to imperfect 
particle separation shrink exponentially with distance. 
Thus, $V$ needs to scale only logarithmically
with $\epsilon$. Similarly, by the analysis of \sect{qubits},
the number of qubits per site scales only logarithmically with
$\epsilon$. By \eq{hightrotter}, the errors resulting from use of a
$k\th$-order Suzuki-Trotter formula with $n$ timesteps are $\epsilon \sim
n^{-2k}$. Thus, the complexity scales as $\epsilon^{-1/2k}$. For
large $k$, the dominant contributions to scaling with $\epsilon$ are
spatial discretization and imperfect adiabaticity.

The effect of spatial discretization is captured by (infinitely many) 
additional terms in the effective Hamiltonian. Truncation of 
these terms alters the calculated probability 
of scattering events. In particular, the two dominant extra terms in the 
effective Hamiltonian are $\sum_i \phi \partial_i^4 \phi$  and $\phi^6$
terms, arising from discretization of $(\nabla_a \phi)^2$ and 
quantum effects, respectively. The coefficient of the
$\sum_i \phi \partial_i^4 \phi$ term is $O(a^2)$, and the coefficient
of the $\phi^6$ term is $O(a^{5-d})$, so that the former dominates for
$d=1,2$, whereas the latter makes a comparable contribution for
$d=3$. Thus, the overall discretization error is
\begin{equation}
\label{spatial_contrib}
\epsilon = O(a^2) \,, \quad d=1,2,3 \,.
\end{equation}
(To improve the scaling, one can use better finite differences 
to approximate the derivative, and/or include the $\phi^6$ operator.
However, renormalization and mixing of the coefficients make this idea
more complicated than it is in standard numerical analysis.)

The diabatic errors at weak coupling are estimated and summarized in
\sect{weak}. The errors are quantified by a probability $\epsilon$ of
observing stray particles. Substituting the $a \sim \sqrt{\epsilon}$
dependence from \eq{spatial_contrib} into \eq{Gstrict}
yields\footnote{Whether we use \eq{Gstrict} or \eq{Glenient}
  affects only the scaling with $V$.}
\begin{equation}
G_{\mathrm{adiabatic}} \sim \left( \frac{1}{\epsilon}
\right)^{1+d/2+o(1)} \,,\quad d=1,2,3 
\end{equation}
scaling for the adiabatic state preparation. We use little-$o$
notation to convey precisely that we are neglecting both logarithmic
factors and contributions to the exponent that become arbitrarily
small as we use higher-order Suzuki-Trotter formulae. The other slow
part of the algorithm is the preparation of the free vacuum. This
scales as
\begin{equation}
G_{\mathrm{prep}} = \tilde{O}(\mathcal{V}^{2.376}) = \tilde{O}(a^{-2.376d}) =
\tilde{O}(\epsilon^{-1.188d}) \,,
\end{equation}
where the last equality follows from \eq{spatial_contrib}. 
Thus, in $d=1$ the adiabatic state preparation is the dominant cost, 
whereas in $d=2,3$ the preparation of the free vacuum dominates. This leaves 
a final asymptotic scaling of
\begin{equation}
G_{\mathrm{total}} = O(G_{\mathrm{adiabatic}} + G_{\mathrm{prep}}) = 
\left\{ \begin{array}{ll}
\left( \frac{1}{\epsilon} \right)^{1.5+o(1)} \,, & d=1, \\
\left( \frac{1}{\epsilon} \right)^{2.376+o(1)} \,, & d=2, \\
\left( \frac{1}{\epsilon} \right)^{3.564+o(1)} \,, & d=3.
\end{array} \right.
\end{equation}

The number of quantum gates used to simulate the
strongly coupled theory has scaling in $1/(\lambda_c - \lambda_0)$ and
$p$ that is dominated by adiabatic state preparation (\sect{strong}).
We also estimate scaling with $n_{\mathrm{out}}$ as follows. For
two incoming particles with momenta $\mathbf{p}$ and $\mathbf{-p}$, the
maximum number of kinematically allowed outgoing particles is $n_{\mathrm{out}}
\sim p$. For continuum behavior, $p = \eta/a$ for constant $\eta \ll
1$. Furthermore, one needs $V \sim n_{\mathrm{out}}$ to obtain
good asymptotic out states separated by a distance of at least $\sim
1/m_0$.  Thus, $\mathcal{V} \sim n_{\mathrm{out}}^{d+1}$, so one needs
$n_{\mathrm{out}}^{2.376 (d+1)}$ gates to prepare the free 
vacuum and, by \eq{strongpscale}, $n_{\mathrm{out}}^{2d+3+o(1)}$ gates to
reach the interacting theory adiabatically. (The adiabatic
turn-off takes no longer than the adiabatic turn-on.)
Hence the total scaling in $n_{\mathrm{out}}$ is dominated by
preparation of the free vacuum in three-dimensional spacetime, but by
adiabatic turn-on in two-dimensional spacetime. These results are
summarized in Table~\ref{strongtable}.

\subsection{Mass Renormalization}
\label{QPT}

The physical, or renormalized, mass as a function of the coupling features 
prominently in our calculations. For the weak-coupling regime, its form is 
obtained by perturbation theory. For the strong-coupling
regime, we use its known behavior near the phase transition.

At first order in the coupling, the shift of the squared mass is given by 
$i$ times the one-loop Feynman diagram
\begin{equation} \label{diag1}
\includegraphics[width=1.2in]{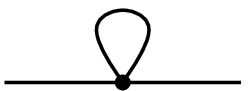} \,. 
\end{equation}
At second order, there is also a contribution from the two-loop diagram
\begin{equation} \label{diag2}
\includegraphics[width=1.2in]{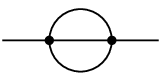} \,. 
\end{equation}
The calculation of these diagrams is quite analogous to standard
calculations in perturbative quantum field theory, but there 
are a couple of differences. First, the propagator is different because
of the discretization. Secondly, integrals over components $1,\ldots, d$
(but not component $0$) of loop momenta are cut off by $\pi/a$, that is, 
the lattice spacing acts as an ultraviolet regulator. These differences 
alter the nature the integrals and hence what methods can be used to 
evaluate them.  


The existence of a phase transition in the $\phi^4$ theory in
$D=2$ or $3$ spacetime dimensions was shown rigorously
in \cite{Glimm:1974tz,Guerra:1975ym,McBryan:1976ga}.
As the system approaches it, thermodynamic functions and correlation 
functions exhibit power-law behavior, as is characteristic of
a second-order phase transition. In particular, for constant $m_0^2$, 
\begin{equation}
\label{numass}
m \sim |\lambda_0 - \lambda_c|^\nu \,,
\end{equation}
where $\lambda_c$, the critical value of the coupling, depends on $m_0^2$. 

Empirically, it has been found that systems with second-order phase
transitions can be classified into universality classes.
Within each class, critical exponents are universal, taking the same
values for all systems.
(This universality is explained by the concept of the renormalization
group.) The $\phi^4$ theory is believed to be in the same
universality class as the Ising model, for which
\begin{equation}
\label{nu}
\nu = \left\{ \begin{array}{ll} 1 \,, & D=2 \,,
\\
0.63\ldots \,, & D=3 \,.
\end{array} \right.
\end{equation}
The value above for $D=3$ has also been obtained directly in the $\phi^4$
theory by Borel resummation \cite{LeGuillou:1977ju}.


In $D=4$ dimensions, in contrast, the believed triviality of the 
continuum $\phi^4$ theory implies that there is no non-trivial fixed point
of the renormalization group and hence no phase transition as one varies
($m_0^2$, $\lambda_0$). Moreover, triviality places bounds on the 
maximum value of the renormalized coupling \cite{Luscher:1987ay}.
In particular, strong coupling requires $p a$ to be $O(1)$: 
in the continuum-like regime, renormalized perturbation theory
should be valid.

\subsection{Representation by Qubits}
\label{qubits}

The required number of qubits per site is
\begin{equation}
\label{placevalue}
n_b = \log \left( 1+2 \lfloor \phi_{\max}/\delta_{\phi} \rfloor \right)
\,.
\end{equation}
In this section we show that one can simulate processes at
energy scale $E$, while maintaining $1-\epsilon$ fidelity to the exact
state, with $n_b$ logarithmic in $1/a$, $1/\epsilon$, and $V$. Our
analysis is nonperturbative, and thus applies equally to strongly and
weakly coupled $\phi^4$ theory. 

Let $\ket{\psi}$ be the state, expressed in the field representation, 
namely,
\begin{equation}
\ket{\psi} = \int_{-\infty}^\infty d \phi_1 \ldots
\int_{-\infty}^\infty d \phi_{\mathcal{V}} \ 
\psi(\phi_1,\ldots,\phi_{\mathcal{V}}) 
\ket{\phi_1, \ldots, \phi_{\mathcal{V}}} \,,
\end{equation}
and let
\begin{equation}
\ket{\psi_{\mathrm{cut}}} = \int_{-\phi_{\max}}^{\phi_{\max}} d \phi_1
\ldots \int_{-\phi_{\max}}^{\phi_{\max}} d \phi_{\mathcal{V}} \ 
  \psi(\phi_1, \ldots, \phi_{\mathcal{V}}) \ket{\phi_1, \ldots
    \phi_{\mathcal{V}}} \,.
\end{equation}
Then
\begin{equation}
\braket{\psi}{\psi_{\mathrm{cut}}} =  \int_{-\phi_{\max}}^{\phi_{\max}} d \phi_1
\ldots \int_{-\phi_{\max}}^{\phi_{\max}} d \phi_{\mathcal{V}}
\ \rho(\phi_1, \ldots, \phi_{\mathcal{V}}) \,,
\end{equation}
where $\rho$ is the probability distribution
\begin{equation}
\rho(\phi_1,\ldots, \phi_{\mathcal{V}}) = | \psi(\phi_1, \ldots, 
\phi_{\mathcal{V}})|^2 \,.
\end{equation}
In other words, $\braket{\psi}{\psi_{\mathrm{cut}}} = 1 -
p_{\mathrm{out}}$, where $p_{\mathrm{out}}$ is the probability that at
least one of $\phi_1,\ldots,\phi_{\mathcal{V}}$ is out of the range
$[-\phi_{\max},\phi_{\max}]$. By the union bound
($ \mathrm{Pr}(A \cup B) \leq \mathrm{Pr}(A) + \mathrm{Pr}(B)$),
\begin{equation}
\braket{\psi}{\psi_{\mathrm{cut}}} \geq 1 - \mathcal{V}
\max_{\mathbf{x} \in \Omega} p_{\mathrm{out}}(\mathbf{x}) \,,
\end{equation}
where $p_{\mathrm{out}}(\mathrm{x})$ is the probability that
$\phi(\mathbf{x})$ is out of the range $[-\phi_{\max},\phi_{\max}]$.

Let $\mu_{\phi(\mathbf{x})}$ and $\sigma_{\phi(\mathbf{x})}$ denote
the mean and standard deviation of $\phi(\mathbf{x})$ determined by
$\rho$. By Chebyshev's inequality, choosing $\phi_{\max} =
\mu_{\phi(\mathbf{x})} + c \sigma_{\phi(\mathbf{x})}$ ensures
\begin{equation}
p_{\mathrm{out}}(\mathbf{x}) \leq \frac{1}{c^2} \,.
\end{equation}
Thus, choosing
\begin{equation}
\label{phichoice}
\phi_{\max} = O \left( \max_{\mathbf{x} \in \Omega} \left(
\mu_{\phi(\mathbf{x})} + \sqrt{\frac{\mathcal{V}}{\epsilon}}
\sigma_{\phi(\mathbf{x})} \right) \right) 
\end{equation}
ensures $\braket{\psi}{\psi_{\mathrm{cut}}} \geq 1-\epsilon$.

Next, we observe the following, which is straightforward to prove.

\begin{proposition}
\label{canonfourier}
Let $\hat{p}$ and $\hat{q}$ be Hermitian operators on
$L^2(\mathbb{R})$ obeying the canonical commutation relation
$[\hat{p},\hat{q}]=i \id$. Then the eigenbasis of $\hat{p}$ is the
Fourier transform of the eigenbasis of $\hat{q}$.
\end{proposition}

By Proposition~\ref{canonfourier}, the eigenbasis of $a^d \pi(\mathbf{x})$ 
is the Fourier transform of the eigenbasis of $\phi(\mathbf{x})$. Thus,
discretizing $\phi(\mathbf{x})$ in increments of
$\delta_{\phi(\mathbf{x})}$ is roughly equivalent to the truncation $-
\pi_{\max} \leq \pi(\mathbf{x}) \leq \pi_{\max}$, where
\begin{equation}
\label{pimax}
\pi_{\max} = \frac{1}{a^d \delta_{\phi(\mathbf{x})}} \,.
\end{equation}
By the same argument used to choose $\phi_{\max}$, choosing
\begin{equation}
\label{pichoice}
\pi_{\max} = O \left( \max_{\mathbf{x} \in \Omega} \left(
\mu_{\pi(\mathbf{x})} + \sigma_{\pi(\mathbf{x})}
\sqrt{\frac{\mathcal{V}}{\epsilon}} \right) \right)
\end{equation}
ensures fidelity $1-\epsilon$ between $\ket{\psi}$ and its truncated
and discretized version.

To obtain useful bounds on $\phi_{\max}$ and $\pi_{\max}$, we must bound
$\mu_{\phi(\mathbf{x})}$, $\sigma_{\phi(\mathbf{x})}$, 
$\mu_{\pi(\mathbf{x})}$, and $\sigma_{\pi(\mathbf{x})}$. 
To this end, we make the following straightforward observation.
\begin{proposition}
\label{moments}
Let $M$ be a Hermitian operator and let $\ket{\psi}$ be a quantum
state. Then $|\bra{\psi} M \ket{\psi}| \leq \sqrt{ \bra{\psi} M^2
  \ket{\psi} }$.
\end{proposition}

\begin{proof}
For brevity, let $\langle Q \rangle = \bra{\psi} Q \ket{\psi}$ for any
observable $Q$. The operator $\left( M - \langle M \rangle \id
\right)^2$ is positive semidefinite. Thus,
\begin{eqnarray}
0 & \leq & \left\langle \left( M - \langle M \rangle \id \right)^2
\right\rangle \\
& = & \left\langle M^2 - 2 \langle M \rangle M + \langle M \rangle^2 \id
\right\rangle \\
& = & \langle M^2 \rangle - \langle M \rangle^2 \,.
\end{eqnarray}
\end{proof} 

\noindent
Applied to the definitions
\begin{eqnarray}
\mu_{\phi(\mathbf{x})} & = & \bra{\psi} \phi(\mathbf{x}) \ket{\psi} \,, \\
\sigma_{\phi(\mathbf{x})} & = & \sqrt{ \bra{\psi} \phi(\mathbf{x})^2
    \ket{\psi} - \bra{\psi} \phi(\mathbf{x}) \ket{\psi}^2} \,, \\
\mu_{\pi(\mathbf{x})} & = & \bra{\psi} \pi(\mathbf{x}) \ket{\psi} \,, \\
\sigma_{\pi(\mathbf{x})} & = & \sqrt{ \bra{\psi} \pi(\mathbf{x})^2
  \ket{\psi} - \bra{\psi} \pi(\mathbf{x}) \ket{\psi}^2} \,,
\end{eqnarray}
Proposition~\ref{moments} implies that $\mu_{\phi(\mathbf{x})}$ and
$\sigma_{\phi(\mathbf{x})}$ are each at most 
$\sqrt{\bra{\psi} \phi(\mathbf{x})^2 \ket{\psi}}$, and
$\mu_{\pi(\mathbf{x})}$ and $\sigma_{\pi(\mathbf{x})}$ are each at
most $\sqrt{\bra{\psi} \pi(\mathbf{x})^2 \ket{\psi}}$. Thus, by
\eq{phichoice} and \eq{pichoice},
\begin{eqnarray}
\phi_{\max} & = & O \left( \max_{\mathbf{x} \in \Omega} \sqrt{
  \frac{\mathcal{V}}{\epsilon} \bra{\psi} \phi(\mathbf{x})^2
  \ket{\psi}} \right) \,,\\
\pi_{\max} & = & O \left( \max_{\mathbf{x} \in \Omega} \sqrt{
  \frac{\mathcal{V}}{\epsilon} \bra{\psi} \pi(\mathbf{x})^2
  \ket{\psi}} \right) \,,
\end{eqnarray}
so that, by \eq{placevalue} and \eq{pimax},
\begin{equation}
n_b = O \left( \log \left( a^d \frac{\mathcal{V}}{\epsilon}
\max_{\mathbf{x},\mathbf{y} \in \Omega} \sqrt{\bra{\psi} \pi(\mathbf{x})^2
\ket{\psi} \bra{\psi} \phi(\mathbf{y})^2 \ket{\psi}} \right) \right) \,.
\end{equation}

To establish logarithmic scaling of $n_b$, we need only prove
polynomial upper bounds on $\bra{\psi} \phi(\mathbf{x})^2 \ket{\psi}$
and $\bra{\psi} \pi(\mathbf{x})^2 \ket{\psi}$. Rather than making a physical
estimate of these expectation values, we prove simple
upper bounds that are probably quite loose. In the adiabatic state
preparation described in \sect{preparing}, the parameters
$m_0^2$ and $\lambda_0$ are varied. The following two propositions
cover all the combinations of parameters used in the adiabatic
preparation and subsequent scattering of both strongly and weakly
coupled wavepackets.

\begin{proposition}
\label{mpos}
Let $H$ be of the form shown in \eq{HS}. Suppose $m_0^2 > 0$ and
$\lambda_0 \geq 0$. Let $\ket{\psi}$ be any state of the field such
that $\bra{\psi} H \ket{\psi} \leq E$. Then $\forall \mathbf{x} \in
\Omega$,
\begin{eqnarray}
\label{phiboundpos}
\bra{\psi} \phi(\mathbf{x})^2 \ket{\psi} & \leq & \frac{2E}{a^d m_0^2} \,,\\
\label{piboundpos}
\bra{\psi} \pi(\mathbf{x})^2 \ket{\psi} & \leq & \frac{2E}{a^d} \,.
\end{eqnarray}
\end{proposition}

\begin{proof}
\begin{eqnarray}
E & \geq & \bra{\psi} H \ket{\psi} \\
  & = & 
\label{almost}
  \bra{\psi} \sum_{\mathbf{x} \in \Omega} a^d \left[ \frac{1}{2}
  \pi(\mathbf{x})^2 + \frac{1}{2} (\nabla_a \phi)^2(\mathbf{x}) +
  \frac{m_0^2}{2} \phi(\mathbf{x})^2 + \frac{\lambda_0}{4!}
  \phi(\mathbf{x})^2 \right] \ket{\psi} \\
& \geq & \bra{\psi} a^d \frac{m_0^2}{2} \phi(\mathbf{x})^2 \ket{\psi},
\end{eqnarray}
where the last inequality follows because all of the operators we have
dropped are positive semidefinite. This establishes
\eq{phiboundpos}. Similarly, we can drop all but the
$\pi(\mathbf{x})$ term from the right-hand side of \eq{almost},
leaving
\begin{equation}
E \geq \bra{\psi} a^d \frac{1}{2} \pi(\mathbf{x})^2 \ket{\psi} \,,
\end{equation}
which establishes \eq{piboundpos}.
\end{proof}

\begin{proposition}
\label{mneg}
Let $H$ be of the form shown in \eq{HS}. Suppose $m_0^2 \leq 0$ and
$\lambda_0 > 0$. Let $\ket{\psi}$ be any state of the field such that
$\bra{\psi} H \ket{\psi} \leq E$. Then $\forall \mathbf{x} \in
\Omega$,
\begin{eqnarray}
\label{phiboundneg}
\bra{\psi} \phi(\mathbf{x})^2 \ket{\psi} & \leq &
-\frac{24m_0^2}{\lambda_0} + \sqrt{\frac{36
    m_0^4}{\lambda_0^2}+\frac{24}{\lambda_0 a^d} \left(
  E + \frac{3(V-a^d)m_0^4}{2\lambda_0} \right)} \,,
\\
\label{piboundneg}
\bra{\psi} \pi(\mathbf{x})^2 \ket{\psi} & \leq & \frac{2}{a^d} \left(
E + \frac{3V m_0^4}{2 \lambda_0} \right) \,,
\end{eqnarray}
where $V$ is the physical volume.
\end{proposition}

\begin{proof}
The operator
\begin{equation}
U(\mathbf{x}) = \frac{m_0^2}{2} \phi(\mathbf{x})^2 +
\frac{\lambda_0}{4!} \phi(\mathbf{x})^4
\end{equation}
is sufficiently simple that we can directly calculate its minimal
eigenvalue $U_{\min}$. If $m_0^2 \leq 0$ and $\lambda > 0$, then
\begin{equation}
\label{Vmin}
U_{\min} = - \frac{3 m_0^4}{2 \lambda_0} \,.
\end{equation}
Thus, for \emph{any} state $\ket{\psi}$,
\begin{equation}
\label{anystate}
\bra{\psi} \sum_{\mathbf{x} \in \Omega} a^d U(\mathbf{x}) \ket{\psi}
\geq \frac{-3Vm_0^4}{2 \lambda_0} \,.
\end{equation}
Hence, recalling \eq{HS}, we obtain
\begin{eqnarray}
E & \geq & \bra{\psi} H \ket{\psi} \\
& = & 
\label{beginning}
\bra{\psi} \sum_{\mathbf{x} \in \Omega} a^d \left[ \frac{1}{2}
  \pi(\mathbf{x})^2 + \frac{1}{2} ( \nabla_a \phi)^2(\mathbf{x}) +
  \frac{m_0^2}{2} \phi(\mathbf{x})^2 + \frac{\lambda_0}{4!}
  \phi(\mathbf{x})^4 \right] \ket{\psi} \\
& \geq &
\label{secondtolast}
 \bra{\psi} \sum_{\mathbf{x} \in \Omega} a^d \left[
  \frac{1}{2} \pi(\mathbf{x})^2 + \frac{1}{2} (\nabla_a
  \phi)^2(\mathbf{x}) \right] \ket{\psi} - \frac{3V m_0^4}{2 \lambda_0}
\\
& \geq & 
\label{last}
\bra{\psi} \frac{a^d}{2} \pi(\mathbf{x})^2 \ket{\psi} -
\frac{3Vm_0^4}{2 \lambda_0} \,.
\end{eqnarray}
\eq{secondtolast} follows from \eq{anystate}. \eq{last} holds
(for any choice of $\mathbf{x}$) because all of the operators we have
dropped are positive semidefinite. This establishes
\eq{piboundneg}.

Similarly, dropping positive operators from \eq{beginning} and
using \eq{anystate} yield, for any $\mathbf{x}$,
\begin{equation}
a^d \bra{\psi} \left( \frac{m_0^2}{2} \phi(\mathbf{x})^2 +
\frac{\lambda_0}{4!} \phi(\mathbf{x})^4 \right) \ket{\psi} \leq
\left( E + \frac{3(V-a^d)m_0^4}{2\lambda_0} \right) \,.
\end{equation}
Applying Proposition \ref{moments} with $M = \phi(\mathbf{x})^2$ shows
that $\bra{\psi} \phi(\mathbf{x})^4 \ket{\psi} \geq \bra{\psi}
  \phi(\mathbf{x})^2 \ket{\psi}^2$. Thus,
\begin{equation}
a^d \left[ \frac{m_0^2}{2} \bra{\psi} \phi(\mathbf{x})^2 \ket{\psi} +
\frac{\lambda_0}{4!} \bra{\psi} \phi(\mathbf{x})^2 \ket{\psi}^2
\right] \leq \left( E + \frac{3(V-a^d)m_0^4}{2\lambda_0} \right) \,.
\end{equation}
Via the quadratic formula, this implies \eq{phiboundneg}.
\end{proof}


\subsection{Adiabatic Preparation of Interacting Wavepackets}
\label{preparing}

In this section, we analyze the adiabatic state-preparation procedure.
To analyze the error due to finite
$\tau$ and $J$, we consider the process of preparing a single-particle
wavepacket. The procedure
performs similarly in preparing wavepackets for multiple particles
provided the particles are separated by more than the characteristic
length $1/m$ of the interaction.

The phase induced by $M_j$ on the momentum-$p$ eigenstate of $H(s)$
(with energy $E_p(s)$) is
\begin{equation}
\theta_j(p) = \left( E_p \left( \frac{j+1}{J} \right) + E_p \left(
\frac{j}{J} \right) \right) \frac{\tau}{2J} - \tau
\int_{j/J}^{(j+1)/J} ds E_p(s) \,.
\end{equation}
Taylor expanding $E_p$ about $s=(j+\frac{1}{2})/J$ yields
\begin{equation}
\label{thetaj}
\theta_j(p) = \frac{\tau}{12 J^3} \frac{\partial^2 E_p}{\partial s^2} 
+ O(J^{-5}) \,.
\end{equation}
Thus the total phase induced is
\begin{eqnarray}
\theta(p) & = & \sum_{j=0}^{J-1} \theta_j(p) \\
& \simeq & \frac{\tau}{12 J^2} \int_0^1 ds \frac{\partial^2
  E_p}{\partial s^2} \\
& = & \frac{\tau}{12 J^2} \left. \frac{\partial E_p}{\partial s}
\right|_0^1 \,, \label{thetap1}
\end{eqnarray}
where the approximation holds for large $J$. For a Lorentz-invariant
theory, $E_p(s)$ must take the form
\begin{equation}
\label{LI}
E_p(s) = \sqrt{p^2+m^2(s)} \,.
\end{equation}
This should be a good approximation for the lattice theory provided
the particle momentum satisfies $p \ll 1/a$. Substituting \eq{LI}
into \eq{thetap1} yields
\begin{equation}
\label{thetap2}
\theta(p) \simeq \frac{\tau}{24 J^2} \left. \frac{\frac{\partial
    m^2}{\partial{s}}}{\sqrt{p^2 + m^2(s)}} \right|_0^1 \,.
\end{equation}

Next, we consider the effect of this phase shift on a wavepacket
centered around momentum $\bar{p}$. If the wavepacket is narrowly
concentrated in momentum, then we can Taylor expand $\theta(p)$ to
first order about $\bar{p}$:
\begin{equation}
\theta(p) \simeq \theta(\bar{p}) + \mathcal{D} \cdot (p-\bar{p}) \,,
\end{equation}
where
\begin{equation}
\mathcal{D} = \left. \frac{\partial \theta}{\partial p}
\right|_{\bar{p}} \,. \label{D}
\end{equation}
The phase shift $e^{i \mathcal{D} \cdot (p - \bar{p})}$ induces a 
translation (in position space) of any wavepacket by a distance $\mathcal{D}$. 
(The second-order term in the Taylor expansion induces
broadening.) From \eq{D} and \eq{thetap2}, we have
\begin{equation}
\label{phasecrit}
\mathcal{D} \simeq \left| \frac{\tau |\bar{p}| }{24 J^2}
\left. \frac{\frac{\partial m^2} {\partial  s}}{\left( \bar{p}^2 +
  m^2(s) \right)^{3/2}} \right|_{s=0}^{s=1} \right| \,.
\end{equation}

We next determine the complexity by demanding that the propagation length
$\mathcal{D}$ be restricted to some small constant, and that the
probability of diabatic particle creation be small. Together, these
criteria determine $J$ and $\tau$. We can obtain a tighter bound in
the perturbative case than in the general case, so we treat these
separately.

\begin{figure}
\begin{center}
\includegraphics[width=0.33\textwidth]{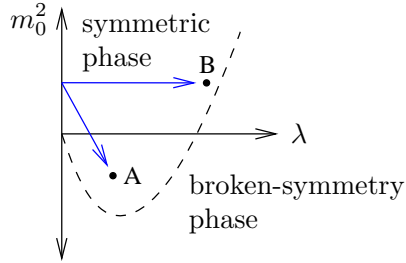}
\caption{\label{paths} The dashed line illustrates schematically the
  location of a quantum phase transition of $\phi^4$ theory in two and
  three spacetime dimensions. A and B denote weakly and strongly coupled
  continuum-like theories, respectively. We prepare them adiabatically
  by following the arrows starting from the massive free theory
  ($m_0^2 >0$, $\lambda_0 = 0$). To maintain adiabaticity, the path
  must not cross the quantum phase transition.}
\end{center}
\end{figure}

\subsubsection{Weak Coupling}
\label{weak}

In the perturbative continuum limit $a \to 0$, $m_0^2$ is
negative. For fixed small $a$, we can adiabatically approach a
perturbative continuum-like theory by taking the straight-line path
depicted in Fig.~\ref{paths}, namely, the following parameterization of
\eq{HS}:
\begin{eqnarray}
m_0^2(s) & = & (m^{(1)})^2 + s \lambda_0 \mu \,, \nonumber \\
\lambda_0(s) & = & s \lambda_0 \label{weakpath} \,.
\end{eqnarray}
Using perturbation theory (see diagram~\ref{diag1}), one finds that it
is particularly efficient to choose 
\begin{equation}
\mu = 
\left\{
\begin{array}{ll}
-\frac{1}{8\pi} \log\Big(\frac{64}{m^2a^2}\Big) 
+ \cdots \,,
 & d=1 \,,\\
[5pt]
-\frac{r_0^{(2)}}{16\pi^2}\frac{1}{a} 
+ \cdots \,, 
& d=2 \,,\\
[5pt]
-\frac{r_0^{(3)}}{32\pi^3}\frac{1}{a^2}
+ \cdots \,, 
& d=3 \,,
\end{array}
\right.
   \label{eq:path}
\end{equation}
so that, at first order in $\lambda_0$, the physical mass remains
fixed at $m^{(1)}$ for all $s$. 
Here, $r_0^{(2)} = 25.379\ldots $ and $r_0^{(3)} = 112.948\ldots $.
In the perturbative regime, this
should ensure that the path does not cross the quantum phase
transition.

To calculate the variation of physical mass with $s$, we must go to
second order in $\lambda_0$ (see diagram~\ref{diag2}). The result is 
\begin{equation}
\label{uptosecond}
 m^2(s) = (m^{(1)} )^2 + s^2 m_2^2 + O(\lambda_0^3)\,,
\end{equation}
where
\begin{equation}
\label{m22}
m_2^2 = \left\{ \begin{array}{ll} O \left( \lambda_0^2/
(m^{(1)})^2 \right) \,, &
d=1 \,,
\vspace{4pt}\\
O\left(\lambda_0^2 \log (m^{(1)}a)\right) \,, & 
d=2 \,,
\vspace{4pt}\\
O(\lambda_0^2/a^2) \,, & 
d=3 \,.
\end{array} \right.
\end{equation}
Substituting \eq{uptosecond} into \eq{phasecrit} yields 
\begin{equation}
\label{phasecrit2}
\frac{\tau |\bar{p}|}{12 J^2} \frac{m_2^2}{\left( \bar{p}^2 +
  (m^{(1)})^2 + m_2^2 \right)^{3/2} } \leq \mathcal{D} \,.
\end{equation}
If we are considering a fixed physical process and using successively
smaller $a$ to achieve higher precision then, by \eq{m22},
it suffices to choose $J$ to scale as
\begin{equation}
\label{Jscale}
J = \left\{ \begin{array}{ll} 
\tilde{O} \left(\sqrt{\frac{m^{(1)} \tau}{\lambda_0 \mathcal{D}}}
\right) \,, & d=1 \,, \vspace{5pt} \\
\tilde{O} \left( \sqrt{\frac{\tau}{\lambda_0 \mathcal{D}}} \right) \,, 
& d=2 \,, \vspace{5pt} \\
\tilde{O} \left(  \sqrt{\frac{a \tau}{\lambda_0 \mathcal{D}}}
\right) \,, & d = 3 \,. \end{array} \right.
\end{equation}
Note that, for $d=3$, $J$ is suppressed by $\sqrt{a}$. This is because,
as $s$ increases, the (uncancelled) two-loop contribution to the physical
mass makes the particle very heavy until $s$ is very close to one.
Hence, the particle propagates slowly, and less backward evolution is
required.

To determine $\tau$, we next consider adiabaticity. Let $H(s)$ be any 
Hamiltonian differentiable with respect to $s$. Let
$\ket{\phi_l(s)}$ be an eigenstate $H(s) \ket{\phi_l(s)} = E_l(s)$
separated by a non-zero energy gap for all $s$. Let $\ket{\psi_l(t)}$
be the state obtained by Schr\"odinger time evolution according to
$H(t/\tau)$ with initial condition $\ket{\psi_l(0)} =
\ket{\phi_l(0)}$. The diabatic transition amplitude to any other eigenstate 
$H(s) \ket{\phi_k(s)} = E_k(s)
\ket{\phi_k(s)}$ ($k \neq l$) is \cite{Messiah}
\begin{equation}
\label{traditional_integral}
\braket{\phi_k(s)}{\psi_l(\tau s)} \sim \int_{0}^{s} d \sigma
\frac{\bra{\phi_k(\sigma)} \frac{dH}{ds} \ket{\phi_l(\sigma)}}{E_l(\sigma)-E_k(\sigma)}
e^{i \tau (\varphi_k(\sigma)-\varphi_l(\sigma))} \left( 1 + O(1/\tau)
\right) \,.
\end{equation}
(The integrand is made well-defined by the phase convention
$\bra{\phi_k} \frac{d \ket{\phi_k}}{ds} = 0$.) Here,
\begin{equation}
\varphi_l(s) = \int_{0}^s d \sigma E_l(\sigma) \,.
\end{equation}
In the case that $E_l$, $E_k$, and $\bra{\phi_k} \frac{dH}{ds}
\ket{\phi_l}$ are $s$-independent, this integral gives
\begin{equation}
\label{traditional}
\braket{\phi_k(s)}{\psi_l(\tau s)} \sim \left( 1 - e^{i \tau
  (E_k-E_l)s} \right) \frac{\bra{\phi_k}
  \frac{dH}{ds} \ket{\phi_l}}{-i \tau (E_k-E_l)^2} (1 + O(1/\tau^2)) \,.
\end{equation}
In the case that these quantities are approximately $s$-independent,
\eq{traditional} should hold as an approximation. 

In reality, we wish to prepare a wavepacket state, not an
eigenstate. However, the wavepacket is well separated from other
particles and narrowly concentrated in momentum space. Thus, we shall
approximate it as an eigenstate $\ket{\phi_l(s)}$. Furthermore, by our
choice of path, the energy gap is kept constant to first order in the
coupling, and thus \eq{traditional} should be a good approximation
to \eq{traditional_integral}.

Summing the transition amplitudes to some state $\ket{\phi_k}$ from
the $J$ steps in in our preparation process, and applying the
triangle inequality\footnote{The $O(J)$ scaling obtained by the
  triangle inequality can be confirmed by a more detailed calculation
  taking into account the relative phases of the contributions to the
  total transition amplitude.} yield the following:
\begin{equation} 
\label{totaldiabatic}
\left| \braket{\phi_k}{\psi_l(\tau)} \right| = O \left(
\frac{1}{\tau} \sum_{j=0}^J \left| \frac{\bra{\phi_k(j/J)} \frac{dH}{ds}
  \ket{\phi_l(j/J)}}{(E_k(j/J) - E_l(j/J))^2} \right| \right) \,.
\end{equation}

The $j=0$ term in this sum can be evaluated exactly, because it arises
from the free theory. At $j \neq 0$ the theory is no longer exactly
solvable. However, one obtains the lowest-order contribution to the matrix
element $\bra{\mathbf{p}_1,\mathbf{p}_2,\mathbf{p}_3,\mathbf{p}_4;s=1} 
\phi^4 \ket{\mathrm{vac}(1)}$ in renormalized perturbation theory
simply by taking the $j=0$ expression and replacing $m_0$
with the physical mass and $\lambda_0$ with the physical coupling. Our
adiabatic path \eq{weakpath} is designed so that the physical mass
at $s=1$ matches the bare mass at $j=0$ (at least to first order in
$\lambda_0$). Furthermore, the physical coupling differs from the bare 
coupling only by a logarithmically divergent (in $a$) correction for 
$d=3$ and non-divergent corrections for $d=1,2$. Thus we can make the 
following approximation:
\begin{equation}
\label{summedadiabatic}
\left| \braket{\phi_k}{\psi_l(\tau)} \right| = \tilde{O} \left(
\frac{J}{\tau} \left| \frac{\bra{\phi_k(0)} \frac{dH}{ds}
  \ket{\phi_l(0)}}{(E_k(0) - E_l(0))^2} \right| \right) \,.
\end{equation}

Diabatic errors come in two types, creation of particles from the
vacuum, and splitting of the incoming particles. The matrix element in the
numerator of \eq{summedadiabatic} can correspondingly be decomposed
as the sum of two contributions. We first consider particle creation
from the vacuum, approximating $\ket{\phi_j(s)}$ as
$\ket{\mathrm{vac}(s)}$. 

By \eq{weakpath},
\begin{equation}
\frac{dH}{ds} = \sum_{\mathbf{x} \in \Omega} a^d \left[
  \frac{\lambda_0}{4!} \phi^4(\mathbf{x}) + \lambda_0 \mu
  \phi^2(\mathbf{x}) \right] \,.
\end{equation}
Substituting this into the numerator of \eq{summedadiabatic},
setting $\ket{\phi_l(0)} = \ket{\mathrm{vac}(0)}$, and expanding
$\phi$ in terms of creation and annihilation operators show that the
only potentially non-zero transition amplitudes are to states
$\ket{\phi_k(0)}$ of two or four particles. The transition amplitude
to states of four particles arise solely from the $\phi^4$ term in
$\frac{dH}{ds}$. The transition amplitude to states of two particles
has contributions from the $\phi^4$ term and the $\phi^2$ term in
$\frac{dH}{ds}$. These actually cancel, because of our choice of $\mu$. 
(Note that this requires tuning of $\mu$.)
At $s=0$, the numerator of \eq{summedadiabatic} is therefore the
following: 
\begin{equation}
\label{treelevel}
\bra{\mathbf{p}_1,\mathbf{p}_2,\mathbf{p}_3,\mathbf{p}_4} 
\frac{\lambda_0}{4!}
\sum_{\mathbf{x} \in \Omega} a^d \phi^4(\mathbf{x})
\ket{\mathrm{vac}(0)} = \frac{ \lambda_0
  \delta_{\mathbf{p}_1+\mathbf{p}_2+\mathbf{p}_3+\mathbf{p}_4,0}}
{4 V \sqrt{\omega(\mathbf{p}_1) \omega(\mathbf{p}_2)
    \omega(\mathbf{p}_3) \omega(\mathbf{p}_4)}} \,.
\end{equation}
We obtain the probability of excitation due to creation of four
particles from the vacuum by squaring the amplitude estimated above,
and then summing over all allowed combinations of the four outgoing
momenta: 
\begin{equation}
P_{\mathrm{create}} \sim 
\sum_{\mathbf{p}_1,\mathbf{p}_2,\mathbf{p}_3,\mathbf{p}_4 \in \Gamma}
\frac{J^2 \lambda_0^2
  \delta_{\mathbf{p}_1+\mathbf{p}_2+\mathbf{p}_3+\mathbf{p}_4,0}}
{V^2 \tau^2 (\omega(\mathbf{p}_1)+\omega(\mathbf{p}_2) +
  \omega(\mathbf{p}_3) + \omega(\mathbf{p}_4))^4 \omega(\mathbf{p}_1)
  \omega(\mathbf{p}_2) \omega(\mathbf{p}_3) \omega(\mathbf{p}_4)} \,.
\end{equation}
This sum is difficult to evaluate exactly; instead, we shall simply
estimate its asymptotic scaling. The question is, with which
parameter should we consider scaling? 
There are at least three regimes in which classical methods for computing
scattering amplitudes break down or are inefficient: strong coupling, 
large numbers of external particles, and high precision. 
In this section we are considering only weak coupling 
(that is, $\lambda/m^{4-D} \ll 1$), leaving discussion of strong coupling
until the next section. 
For an asymptotically large number of external particles, the efficiency
of our algorithm depends upon strong coupling, for the following reason.
A connected Feynman diagram involving $n$ external particles must have 
at least $v=O(n)$ vertices, so the amplitude for such a process
is suppressed by a factor of $\left( \frac{\lambda}{E^{4-D}}
\right)^v$, where $E$ is the energy scale of the process. Since 
$E \geq m$, many-particle scattering events are exponentially rare
at weak coupling, and thus cannot be efficiently observed in
experiments or simulations. This leaves the high-precision
frontier. Recall that the perturbation series used in quantum field 
theory are asymptotic but not convergent. Thus, perturbative methods 
cannot be extended to arbitrarily high precision. 

Hence, in this section we consider the quantum gate complexity of 
achieving arbitrarily high precision. 
To do so, one chooses $a$ small to obtain
small discretization errors, $V$ large to obtain better
particle separation, $\tau$ long to improve adiabaticity, and $J$
large enough to limit unwanted particle propagation as the interaction
is turned on. Thus, we wish to know the scaling
of $P_{\mathrm{create}}$ with $a$, $\tau$, $V$, and $J$. In this
context, we consider 
$m$, $\lambda$, and $|\mathbf{p}_1|$ to be constants.

We now estimate the scaling of $P_{\mathrm{create}}$ as $a \to 0$.
\begin{eqnarray}
P_{\mathrm{create}} & \sim & \frac{J^2}{V^2 \tau^2}
\sum_{\mathbf{p}_1,\mathbf{p}_2,\mathbf{p}_3  \in \Gamma} \nonumber \\
&  & \frac{\lambda_0^2}
{(\omega(\mathbf{p}_1)+\omega(\mathbf{p}_2)+\omega(\mathbf{p}_3) +
  \omega(-\mathbf{p}_1-\mathbf{p}_2-\mathbf{p}_3))^4
  \omega(\mathbf{p}_1) \omega(\mathbf{p}_2) \omega(\mathbf{p}_3)
 \omega(-\mathbf{p}_1-\mathbf{p}_2-\mathbf{p}_3)} \nonumber \\
& \simeq & \frac{3 J^2}{V^2 \tau^2} 
\sum_{\substack{\mathbf{p}_1,\mathbf{p}_2,\mathbf{p}_3 \in \Gamma \\ |\mathbf{p}_1| > |\mathbf{p}_2|,|\mathbf{p}_3|}} \nonumber \\
&  &
\frac{\lambda_0^2}
 {(\omega(\mathbf{p}_1)+\omega(\mathbf{p}_2)+\omega(\mathbf{p}_3)+\omega(-\mathbf{p}_1-\mathbf{p}_2-\mathbf{p}_3))^4
  \omega(\mathbf{p}_1) \omega(\mathbf{p}_2) \omega(\mathbf{p}_3) \omega(-\mathbf{p}_1-\mathbf{p}_2-\mathbf{p}_3)} \nonumber \\
& \sim & \frac{J^2}{V^2 \tau^2} 
\sum_{\substack{\mathbf{p}_1,\mathbf{p}_2,\mathbf{p}_3 \in \Gamma \\ |\mathbf{p}_1| > |\mathbf{p}_2|,|\mathbf{p}_3|}}
\frac{\lambda_0^2}
 {\omega(\mathbf{p}_1)^6 \omega(\mathbf{p}_2) \omega(\mathbf{p}_3)} \nonumber \\
& \leq & \frac{J^2}{V^2 \tau^2}
 \sum_{\mathbf{p}_1,\mathbf{p}_2,\mathbf{p}_3 \in \Gamma} 
\frac{\lambda_0^2}
 {\omega(\mathbf{p}_1)^6 \omega(\mathbf{p}_2) \omega(\mathbf{p}_3)} \nonumber \\
& \simeq & \frac{V J^2}{\tau^2} \int_{\Gamma} d^d p_1 \int_{\Gamma}
  d^d p_2 \int_{\Gamma} d^d p_3 \frac{\lambda_0^2}{\omega(\mathbf{p}_1)^6
  \omega(\mathbf{p}_2) \omega(\mathbf{p}_3)} \nonumber \\
& = & 
\label{pcreate}
\left\{ \begin{array}{ll} \tilde{O} \left(
  \frac{V J^2}{\tau^2} \right) \,, & d=1,2 \,, \vspace{5pt} \\
\tilde{O} \left( \frac{V J^2}{\tau^2 a} \right) \,, & d=3 \,.
\end{array} \right.
\end{eqnarray}

By \eq{Jscale} and \eq{pcreate}, 
\begin{equation}
P_{\mathrm{create}}= \begin{array}{ll}
\tilde{O} \left( \frac{V}{\tau} \right) \,, & d=1,2,3 \,.
\end{array}
\end{equation}

Next, we consider the process in which the time dependence of the
$\phi^4$ term causes a single particle to split into three. For this
process, the relevant matrix element is
\begin{equation}
\bra{\mathbf{p}_2,\mathbf{p}_3,\mathbf{p}_4} \frac{\lambda_0}{4 !} 
\sum_{\mathbf{x} \in \Omega} 
a^d \phi^4(\mathbf{x}) \ket{\mathbf{p}_1} = \frac{ \lambda_0
  \delta_{\mathbf{p}_2+\mathbf{p}_3+\mathbf{p}_4,\mathbf{p}_1}}
{4 V \sqrt{\omega(\mathbf{p}_1) \omega(\mathbf{p}_2)
    \omega(\mathbf{p}_3) \omega(\mathbf{p}_4)}} \,,
\end{equation}
where $\mathbf{p}_1$ is the momentum of the incoming particle. 
By our choice of path, the physical mass is $s$-independent to
first order in the coupling, and the $s$ dependence of the 
coupling is only logarithmically divergent as $a \to 0$. Thus,
by \eq{totaldiabatic},
\begin{equation}
\label{Psplitdef}
P_{\mathrm{split}} \sim \frac{J^2}{\tau^2 V^2}
\sum_{\mathbf{p}_2,\mathbf{p}_3,\mathbf{p}_4 \in \Gamma}
\frac{\lambda_0^2 \delta_{\mathbf{p}_2+\mathbf{p}_3+\mathbf{p}_4,\mathbf{p}_1}}
{(\omega(\mathbf{p}_2)+\omega(\mathbf{p}_3)+\omega(\mathbf{p}_4)-\omega(\mathbf{p}_1))^4 
\omega(\mathbf{p}_1) \omega(\mathbf{p}_2) \omega(\mathbf{p}_3) \omega(\mathbf{p}_4)} \,.
\end{equation}

Let us now examine the divergence structure of $P_{\mathrm{split}}$ as
$a \to 0$. In the limit of large volume, the sum converges to the
following integral:
\begin{equation}
\frac{2J^2}{\tau^2} \int_\Gamma d^d p_2
\int_\Gamma d^d p_3 \frac{\lambda_0^2}{(\omega(\mathbf{p}_2)+\omega(\mathbf{p}_3) +
  \omega(\mathbf{p}_1-\mathbf{p}_2-\mathbf{p}_3) - \omega(\mathbf{p}_1))^4 \omega(\mathbf{p}_1) \omega(\mathbf{p}_2)
  \omega(\mathbf{p}_3) \omega(\mathbf{p}_1-\mathbf{p}_2-\mathbf{p}_3)} \,. \\
\end{equation}
If this were divergent as $a \to 0$, then by approximating the integrand 
with its value at large
$|\mathbf{p}_2|$ and $|\mathbf{p}_3|$, we would be able to isolate the
divergence:
\begin{equation}
P_{\mathrm{split}} \sim \frac{J^2 \lambda_0^2}{\tau^2
  \omega(\mathbf{p}_1)} \int_\Gamma d^d p_2 \int_\Gamma d^d p_3
\frac{1}{(|\mathbf{p}_2|+|\mathbf{p}_3|+|\mathbf{p}_2+\mathbf{p}_3|)^4
  |\mathbf{p}_2| |\mathbf{p}_3| |\mathbf{p}_2+\mathbf{p}_3|} \,.
\end{equation}
However, for $d=1,2,3$ this is convergent as $a \to 0$. Thus, recalling 
\eq{Jscale}, we obtain
\begin{equation}
P_{\mathrm{split}} = O \left( \frac{J^2}{\tau^2} \right) =
\left\{ \begin{array}{ll} \tilde{O} \left( \frac{1}{\tau} \right) \,, 
& d=1,2 \,,\\
\tilde{O} \left( \frac{a}{\tau} \right) \,, & d=3 \,.
\end{array} \right.
\end{equation}

We can consider two criteria regarding diabatic particle creation. If
our detectors are localized, we may be able to tolerate a low constant
density of stray particles created during state preparation. This
background is similar to that encountered in experiments, and may not
invalidate conclusions from the simulation. Alternatively, one could adopt a
strict criterion by demanding that, with high probability, not even
one stray particle is created in the volume being simulated during
state preparation. This strict criterion can be quantified by
demanding that the adiabatically produced state has an inner product of at
least $1-\epsilon$ with the exact state. This parameter $\epsilon$ is
thus directly comparable with that used in \sect{qubits}, and
the two sources of error can be added. Applying the strict criterion,
we demand that $P_{\mathrm{split}}$ and $P_{\mathrm{create}}$ each be
of order $\epsilon$, and obtain
\begin{equation}
\tau_{\mathrm{strict}} = \tilde{O} \left(
\frac{V}{\epsilon} \right) \,, \quad d=1,2,3  \,.
\end{equation}
Applying the more lenient criterion that $P_{\mathrm{create}}/V$ and
$P_{\mathrm{split}}$ each be of order $\epsilon$ yields
\begin{equation}
\tau_{\mathrm{lenient}} = \tilde{O} \left( \frac{1}{\epsilon} \right) \,,
\quad d=1,2,3 \,.
\end{equation}
For a $k\th$-order Suzuki-Trotter formula, the asymptotic scaling of the
total number of gates needed for adiabatic state preparation is
$O \left( (\mathcal{V} \tau)^{1+\frac{1}{2k}} \right) = O \left( (V
\tau/a^d)^{1+\frac{1}{2k}}) \right)$. Thus,
\begin{eqnarray}
\label{Gstrict}
G_{\mathrm{adiabatic}}^{\mathrm{strict}} & = & \tilde{O} \left(
\left( \frac{V^2}{a^d \epsilon} \right)^{1+\frac{1}{2k}} \right) \,, \\
\label{Glenient}
G_{\mathrm{adiabatic}}^{\mathrm{lenient}} & = & \tilde{O}
\left( \left( \frac{V}{a^d \epsilon} \right)^{1+\frac{1}{2k}} \right) \,.
\end{eqnarray}

\subsubsection{Strong Coupling}
\label{strong}

In two and three spacetime dimensions, we can obtain a strongly coupled
(that is, nonperturbative) field theory by approaching the 
phase transition (\sect{QPT}). As in the case of weak
coupling, the necessary time for adiabatic state preparation depends
on various physical parameters of the system being simulated,
including the momentum of the incoming particles, the volume, the
strength of the final coupling, the number of spatial dimensions, and
the physical mass. To keep the discussion concise, we restrict our
discussion to the case of ultrarelativistic incoming particles, with
coupling strength close to the critical value. Under these conditions,
the incoming particles can produce a shower of many 
$(n_{\mathrm{out}} \sim p/m)$ outgoing particles. Because of the strong
coupling, perturbation theory is inapplicable, and, even
if it could be used, would take exponential computation in the number of
outgoing particles.

In the strongly coupled case, we vary the Hamiltonian \ref{HS}
with $s$ by keeping the bare mass constant at $m_0$ and setting the
bare coupling to $s \lambda_0$. We choose $\lambda_0$ only slightly
below the critical value $\lambda_c$, so that at $s=1$ the system
closely approaches the phase transition, as illustrated in 
Fig.~\ref{paths}. Examining \eq{thetap2} suggests that we can estimate
phase errors by understanding the behavior of $m^2(s)$ at $s=0$ and
$s=1$, without needing to know exactly what happens in between. 
From \eq{eq:path}, 
\begin{equation}
\left. \frac{d m^2}{d s} \right|_{s=0} = \left\{ \begin{array}{ll} 
\frac{\lambda_0}{8 \pi} \log \left( \frac{64}{m_0^2 a^2} \right) &
d=1 \,, \\
\frac{25.379}{16
  \pi^2} \frac{\lambda_0}{a} & 
d=2 \,,
\end{array} \right.
\end{equation}
and, 
from \eq{numass} and \eq{nu}, 
\begin{equation}
\left. \frac{d m^2}{d s} \right|_{s=1} \sim \left\{
\begin{array}{ll}
-2 (\lambda_c - \lambda_0) & 
d=1 \,, \\
-1.26 (\lambda_c - \lambda_0)^{0.26} & 
d=2 \,.
\end{array} \right.
\end{equation}
Thus, \eq{phasecrit} yields
\begin{equation}
\label{unconcrete}
J = \tilde{O} \left( \sqrt{ \frac{\tau \lambda_0}{a^{d-1} p^2 \mathcal{D}}}
\right) \,, \quad d=1,2 \,,
\end{equation}
under the assumption that $(\lambda_c - \lambda_0)$ is very
small.

The result \ref{thetap2} rests on two approximations, a Taylor
expansion to second order in \eq{thetaj}, and an
approximation of a sum by an integral in \eq{thetap1}. The validity
conditions for these approximations become most stringent at $s=1$, 
where the derivatives of $m^2$ with respect to $s$ become large.  
Working out the $O(J^{-4})$ term in \eq{thetap2} at $s=1$, one
finds that it will be much smaller than the $O(J^{-2})$ term at $s=1$
provided
\begin{equation}
\label{Taylorcrit}
J \gg \frac{1}{\lambda_c - \lambda_0} \,.
\end{equation}
Similarly, higher-order terms in the Taylor expansion are suppressed
by additional powers of $\frac{1}{J(\lambda_c-\lambda_0)}$. The criterion
\ref{Taylorcrit} also suffices to justify the approximation of the sum
by an integral in \eq{thetap1}. 

We must next consider adiabaticity to determine $\tau$. In
the ultrarelativistic limit, the relevant energy gap $\gamma$ is $\sim
\frac{m^2}{p}$. This takes its minimum value at $s=1$, namely,
\begin{equation}
\label{gammamin}
\gamma_{\min} \simeq \left\{ \begin{array}{ll}
\frac{(\lambda_c - \lambda_0)^2}{p} \,,  & d=1 \,, \\
\frac{(\lambda_c - \lambda_0)^{1.26}}{p} \,, & d=2 \,.
\end{array} \right.
\end{equation}
Unlike in the perturbative case, we cannot make a detailed quantitative
analysis, but under the condition~\ref{Taylorcrit}, we should again be able 
to apply the traditional adiabatic criterion and obtain a diabatic
transition amplitude scaling as $\frac{J}{\tau \gamma^2}$. Thus, to
keep the error probability at some small constant $\epsilon$,  we have
\begin{equation}
\label{tauscale}
\tau \sim \frac{J}{\gamma^2 \sqrt{\epsilon}} \,.
\end{equation}

We now consider asymptotic scaling with $p$ for fixed
$\lambda_0$. To achieve continuum-like behavior we need $a \ll
\frac{1}{p}$. Thus \eq{unconcrete} yields
\begin{equation}
\label{pscaleJ}
J \sim \tau^{1/2} p^{(d-3)/2} \,, \quad d=1,2 \,.
\end{equation}
Substituting \eq{Taylorcrit} and \eq{gammamin} into
\eq{tauscale}, we see that we need
\begin{equation}
\label{cond1}
\tau \gtrsim p^2 \,, \quad d=1,2 \,.
\end{equation}
Substituting \eq{pscaleJ} and \eq{gammamin} into
\eq{tauscale}, we see that we need 
\begin{equation}
\label{cond2}
\tau \gtrsim p^{d+1} \,, \quad d=1,2 \,.
\end{equation}
The scaling $\tau = O(p^{d+1})$ for $d=1,2$ suffices to satisfy both
conditions~\ref{cond1} and \ref{cond2}. Thus, by \S \ref{Trotter},
the total number of gates scales as
\begin{eqnarray}
G_{\mathrm{strong}} & = & O((V\tau)^{1+o(1)}p^{d+1+o(1)}) \\
& = & O \left( V^{1+o(1)} p^{2d+2+o(1)} \right) \,, \label{strongpscale}
\end{eqnarray}
for $d=1,2$. 

Next, we consider asymptotic scaling with $(\lambda_c - \lambda_0)$
for fixed $p$. The $J$ scaling as $\sqrt{\tau}$ in
\eq{unconcrete} automatically satisfies the condition
\ref{Taylorcrit}. Thus, we substitute \eq{unconcrete} into
\eq{tauscale}, obtaining
\begin{equation}
\tau \sim \left\{ \begin{array}{ll} \left( \frac{1}{\lambda_c -
    \lambda_0} \right)^8 \,, & d=1 \,, \\
\left( \frac{1}{\lambda_c - \lambda_0} \right)^{5.04} \,, & d=2 \,.
\end{array} \right. 
\end{equation}
Thus, using a $k\th$-order Suzuki-Trotter formula, we obtain
\begin{equation}
\label{stronglambdascale}
G_{\mathrm{strong}} \sim \left\{ \begin{array}{ll} \left( \frac{1}{\lambda_c -
    \lambda_0} \right)^{8 \left( 1+\frac{1}{2k} \right)} \,, & d=1 \,, \\
\left( \frac{1}{\lambda_c - \lambda_0} \right)^{5.04
  \left(1+\frac{1}{2k} \right) } \,, & d=2 \,.
\end{array} \right. 
\end{equation}
Note that one could improve this scaling by choosing a more optimized
adiabatic state-preparation schedule, which slows down as the gap gets
smaller. 

\subsection{Suzuki-Trotter Formulae for Large Lattices}
\label{Trotter}

It appears that, while scaling with $t$ has been thoroughly studied,
little attention has been given to scaling of quantum simulation
algorithms with the number of lattice sites $\mathcal{V}$.
Using a result of Suzuki and elementary Lie algebra theory, we derive
linear scaling provided the Hamiltonian is local.

For any even $k$ and any pair of Hamiltonians $A, B$,
\begin{equation}
\label{hightrotter}
\left( e^{i A \alpha_1 t/n} e^{i B \beta_1 t/n} e^{i A \alpha_2 t/n}
e^{i \beta_2 B t/n} \ldots e^{i A \alpha_r t/n} \right)^n = e^{i (A + B) t} 
+ O(t^{2k+1}/n^{2k}) \,,
\end{equation}
where $r = 1+5^{k/2-1}$ and
$\alpha_1,\ldots,\alpha_r,\beta_1,\ldots,\beta_{r-1}$ are specially
chosen coefficients such that $\sum_{j=1}^r \alpha_j = 1$ and
$\sum_{j=1}^{r-1} \beta_j = 1$ \cite{Suzuki}. Thus, using the $k\th$-order 
Suzuki-Trotter formula (\eq{hightrotter}), one can simulate evolution 
for time $t$ with $O\left( t^{\frac{2k+1}{2k}} \right)$ quantum gates
\cite{Cleve_sim}. To determine the $\mathcal{V}$ scaling, we use the
following standard theorem (cf. the Baker-Campbell-Hausdorff
formula).
\begin{theorem}\label{BCH}
Let $A$ and $B$ be elements of a Lie algebra defined over any field of
characteristic 0. Then $e^{A} e^{B} = e^{C}$, where $C$ is a formal
infinite sum of elements of the Lie algebra generated by $A$ and $B$.
\end{theorem}
$A$ and $B$ generate a Lie algebra by commutation and linear
combination. Thus, without requiring any explicit calculation,
Theorem~\ref{BCH} together with \eq{hightrotter} implies
\begin{equation}
\left( e^{i A \delta_1 t/n} e^{i B \delta_2 t/n} \ldots e^{i A
  \delta_r t/n} \right)^n = e^{i (A + B) t} + \Delta_{2k+1}
t^{2k+1}/n^{2k} + O(n^{-(2k+1)}) \,,
\end{equation}
where $\Delta_{2k+1}$ is a linear combination of nested
commutators. In general, $\| \Delta_{2k+1} \|$ could be as large as 
$\left( \max \left\{ \|A\|, \|B\| \right\} \right)^{2k+1}$. However, 
by the canonical commutation relations, one sees that, for the pair of local 
Hamiltonians $H_{\phi}, H_{\pi}$, $\|\Delta_{2k+1}\| =
O(\mathcal{V})$, for any fixed $k$. Thus, one needs only
$n = O \left( t^{\frac{2k+1}{2k}} \mathcal{V}^{\frac{1}{2k}} 
\right)$.  Recalling the $O(\mathcal{V})$ cost for simulating each 
$e^{i H_{\phi} \delta t}$ or $e^{i  H_{\pi} \delta t}$, one sees that 
the total number of gates scales as $O\left( \left( t \mathcal{V} 
\right)^{1+\frac{1}{2k}} \right)$. Note that this conclusion may be of 
general interest, as it applies to any lattice Hamiltonian for which 
non-neighboring terms commute.

In the case of strong coupling, we care not only about how the number
of gates scales with $\mathcal{V}$, but also about scaling with
$p$. In the presence of high-energy incoming particles, the field can
have large distortions from its vacuum state. For example, if
$\bra{\psi} \phi(\mathbf{x}) \ket{\psi}$ is large, then local terms in
$\Delta_{2k+1} \ket{\psi}$ such as $\pi(\mathbf{x}) \phi(\mathbf{x})^3
\ket{\psi}$ can become large. We can obtain a heuristic upper bound on
this effect by noting that, in the strongly coupled case, $m_0^2 > 0$,
so each local term in $H$ is a positive operator. Thus, if $\bra{\psi}
H \ket{\psi} \leq E$, then the expectation value of each of the local
terms is bounded above by $E$. Using $E$ as a simple estimate of the
maximum magnitude of a local term, we see that $\Delta_{2k+1}
\ket{\psi}$, which is a sum of $O(\mathcal{V})$ terms, each of which
is of degree $2k+1$ in the local terms of $H$, has magnitude at most
$O(\mathcal{V} E^{2k+1})$, or in other words $O(\mathcal{V}
p^{2k+1})$. Recalling that $a$ scales as a small multiple of $1/p$, we
see that $\Delta_{2k+1} \ket{\psi} = O(V p^{2k+1+d})$. Thus, $n =
O(p^{1+(1+d)/2k}t^{1+1/2k})$. Each timestep requires $O(\mathcal{V}) =
O(Vp^d)$ gates to implement. Thus, the overall scaling is
$O(p^{d+1+o(1)} (tV)^{1+o(1)})$ quantum gates to simulate the strongly 
coupled theory at large $p$.


\end{document}